\documentclass[a4paper,draft,11pt]{article}

\usepackage{pb-diagram,amsmath,amssymb}

\usepackage{amsthm}
\usepackage[T2A]{fontenc}
\usepackage[cp1251]{inputenc}
\usepackage[active]{srcltx}

\usepackage{amsthm}

\newtheorem{theorem}{Theorem}

\newtheorem{lemma}[theorem]{Lemma}

\newtheorem{proposition}[theorem]{Proposition}
\newtheorem{corollary}[theorem]{Corollary}

\def\myarrow{\ \hbox to 2em{\leaders
\hbox to 0.5ex{\hss\raise 0.55ex\hbox to 0.3ex{\hrulefill}\hss}
\hfill\,\llap{$>$}}\ }

\title{   Duursma's reduced polynomial  \footnote{
Supported by   Contract 015/2014 and  Contract 144/2015  with the Scientific Foundation of Kliment Ohridski  University of Sofia. } }
\author{
}
\date{      }

\begin{document}
\maketitle

\centerline{\scshape Azniv Kasparian }
\medskip
{\footnotesize
 \centerline{Section of Algebra, Department of Mathematics and Informatics}
   \centerline{Kliment Ohridski University of Sofia}
   \centerline{ 5 James Bouchier Blvd., Sofia 1164, Bulgaria}
    \centerline{ {\bf email:} kasparia@fmi.uni-soifa.bg }}

\medskip

\centerline{\scshape  Ivan Marinov   }
\medskip
{\footnotesize
 \centerline{ Obecto - Boutique Software Development Company  }
   \centerline{ 23 A Dragan Tsankov Blvd. }
   \centerline{  Sofia, Bulgaria }
     \centerline{ {\bf email:} imarinov@obecto.com }
}

\thispagestyle{empty}

\begin{abstract}
 The weight distribution $\{ \mathcal{W}_C^{(w)} \} _{w=0} ^n$ of a linear code $C \subset {\mathbb F}_q^n$   is   put  in an explicit bijective correspondence with   Duursma's reduced polynomial $D_C(t) \in {\mathbb Q}[t]$ of $C$.
We prove that the Riemann Hypothesis Analogue for a  linear code $C$    requires the formal self-duality of $C$ and imposes an upper bound on the cardinality $q$ of the basic field, depending on the dimension and the minimum distance of $C$.
  Duursma's reduced polynomial $D_F(t) \in {\mathbb Z}[t]$ of the function field $F = {\mathbb F}_q(X)$ of a curve $X$ of genus $g$ over ${\mathbb F}_q$ is shown to  provide a generating function $\frac{D_F(t)}{(1-q)(1-qt)} = \sum\limits _{i=0} ^{\infty} \mathcal{B}_i t^{i}$ for the numbers $\mathcal{B}_i$ of the effective divisors of degree $i \geq 0$ of a virtual function field of a curve of genus $g-1$ over ${\mathbb F}_q$.
\end{abstract}

  Let $\overline{{\mathbb F}_q} = \cup _{m=1} ^{\infty} {\mathbb F}_{q^m}$ be the algebraic closure of a finite field ${\mathbb F}_q$ and
 $X / {\mathbb F}_q \subset {\mathbb P}^N ( \overline{{\mathbb F}_q})$  be a smooth irreducible projective curve of genus $g$, defined over ${\mathbb F}_q$.
 Denote by $F = {\mathbb F}_q(X)$ the function field of $X$ over ${\mathbb F}_q$ and choose $n$ different ${\mathbb F}_q$-rational points
  $P_1, \ldots , P_n \in X( {\mathbb F}_q) := X \cap  {\mathbb P}^N({\mathbb F}_q)$.
Suppose that $G$ is an effective divisor of $F$ of degree $2g-2 < \deg G = m < n$, whose support is disjoint from the support
 of $D = P_1 + \ldots + P_n$.
The space $L(G) := H^0 (X, \mathcal{O}_X(G))$ of the global holomorphic sections of the line  bundle, associated with $G$ will be referred to as to the Riemann-Roch space of $G$.
We put $l(G) := \dim _{{\mathbb F}_q} L(G)$ and observe that the evaluation map
\[
\mathcal{E}_D : L(G) \longrightarrow {\mathbb F}_q^n,
\]
\[
\mathcal{E}_D (f) = (f(P_1), \ldots , f(P_n))  \ \ \mbox{  for  } \ \ \forall f \in L(G)
\]
is an ${\mathbb F}_q$-linear embedding.
Its image $C := {\rm im} ( \mathcal{E}_D) = \mathcal{E}_D L(G)$ is known as  an algebraic  geometry code or Goppa code.
The minimum distance of $C$ is $d(C) \geq n-m$.
For an arbitrary $s \in {\mathbb N}$ let $N_s (F) := |X( {\mathbb F}_{q^s})|$ be the number of the ${\mathbb F}_{q^s}$-rational points of $X$.
Then the formal power series
\[
Z_F(t) := \exp \left( \sum\limits _{s=1} ^{\infty} \frac{N_s (F)}{s} t^s \right)
\]
is called the Hasse-Weil zeta function of $F$.
It is well known (cf. Theorem 4.1.11 from \cite{NX}) that
\[
Z_F(t) = \frac{L_F(t)}{(1-t)(1-qt)}
\]
for a polynomial $L_F(t) \in {\mathbb Z}[t]$ of degree $2g$.
We refer to $L_F(t)$ as to the Hasse-Weil polynomial of $F$.

In \cite{D1}, \cite{D2} Duursma introduces the genus of a linear code $C \subset {\mathbb F}_q^n$ as the deviation $g  := n+1 -  k - d$ of its dimension $k := \dim _{{\mathbb F}_q} C$  and minimum distance  $d$ from the equality in Singleton bound.
Let $\mathcal{W}_C^{(w)}$ be the number of the codewords $c \in C$ of weight $d  \leq w \leq n$.
Then
\[
\mathcal{W}_C(x,y) := x^n + \sum\limits _{w=d(C)} ^n  \mathcal{W}_C ^{(w)} x^{n-w} y^w
\]
 is called the homogeneous weight enumerator of $C$.
Denote by $\mathcal{M}_{n,s} (x,y)$ the homogeneous weight enumerator of an MDS-code of length $n$  and minimum distance $s$.
 Put  $g^{\perp}$  for  the  genus  of the dual code $C^{\perp}$ of $C$ and $r:= g + g^{\perp}$.

 \begin{proposition}   \label{LinearCode}
{\rm (Duursma \cite{D2})}  For an arbitrary ${\mathbb F}_q$-linear $[n,k,d]$-code $C$, which is not contained in a co\-or\-di\-na\-te hyperplane $H_i := \{ x \in {\mathbb F}_q ^n \ \ \vert \ \  x_i  = 0 \}$ of ${\mathbb F}_q^n$, there exist uniquely determined rational numbers $a_0, \ldots , a_r \in {\mathbb Q}$, such that the homogeneous weight enumerator
\begin{equation}   \label{MDSExpression}
\mathcal{W}_C (x,y) = a_0 \mathcal{M}_{n,d} (x,y) + a_1 \mathcal{M}_{n,d+1}(x,y) + \ldots + a_r \mathcal{M}_{n, d+r} (x,y)
\end{equation}
of $C$ is the linear combination of the homogeneous weight enumerators $\mathcal{M}_{n, d+i}(x,y)$ of MDS-codes of length $n$ and minimum distance $d+i$ with coefficients $a_i$   and
\begin{equation}    \label{CoefficientSum1}
P_C(1) = \sum\limits _{i=0} ^r a_i = 1.
\end{equation}
The $\zeta$-polynomial $P_C(t) := \sum\limits _{i=0} ^r a_i t^{i}$ of $C$ is uniquely determined by
\begin{equation}  \label{CoefficientFormula}
\frac{\mathcal{W}_C(x,y)  - x^n}{q-1} = {\rm Coeff} _{t^{n-d}} \left( \frac{P_C (t) }{(1-t)(1-qt)} [ y (1-t) + xt] ^n \right),
\end{equation}
where ${\rm Coeff} _{t^{n-d}} (f(t))$ stands for the coefficient of $t^{n-d}$ in a  formal power series $f(t) \in {\mathbb C} [[t]]$.
\end{proposition}

 \begin{proposition}    \label{GeometricDzetaPolynomialAGCode}
 {\rm (Duursma's considerations from  \cite{D1})} Let $X / {\mathbb F}_q \subset {\mathbb P}^N ( \overline{{\mathbb F}_q})$ be a smooth irreducible curve of genus $g$, defined over ${\mathbb F}_q$ and
   $G_1, \ldots , G_h$ be a complete list of effective representatives of the linear equivalence classes of the divisors of $F = {\mathbb F}_q (X)$ of degree $2g-2 < m < n$.
   Assume that there exist $n$ different ${\mathbb F}_q$-rational points $P_1, \ldots , P_n \in X( {\mathbb F}_q)$, such that
    $D = P_1 + \ldots + P_n \in {\rm Div} (F)$ has support ${\rm Supp} (D) \cap {\rm Supp} (G_i) = \{ P_1, \ldots , P_n \} \cap {\rm Supp} (G_i) = \emptyset$
     for $\forall 1 \leq i \leq h$.
  If
  \[
  \mathcal{L} (G_i) = H^0 (X, \mathcal{O}_X ( [G_i]) := \{ f \in F^* \ \  \vert \ \  (f) + G_i \geq 0 \} \cup \{ 0 \}
  \]
  are the Riemann-Roch spaces of $G_i$,
  \[
  \mathcal{E}_D : \mathcal{L}( G_i) \longrightarrow {\mathbb F}_q ^n,
  \]
  \[
  \mathcal{E}_D (f) = ( f(P_1), \ldots , f(P_n)) \ \ \mbox{  for  } \ \ \forall f \in \mathcal{L} (G_i)
  \]
  are the evaluation maps at $D$ and $C_i := \mathcal{E}_D \mathcal{L} (G_i)$ are the corresponding Goppa codes with homogeneous weight enumerators
   $\mathcal{W} _{C_i} (x,y)$, then
   \begin{equation}    \label{RelationDzetaPolynomialsAGCodes}
   \sum\limits _{i=1} ^h \frac{\mathcal{W}_{C_i} (x,y) - x^n}{q-1} = {\rm Coeff} _{t^m} \left( \frac{L_F(t)}{(1-t)(1-qt)} [ y (1-t) + xt]^n \right)
   \end{equation}
for the $\zeta$-polynomial $L_F(t)$ of $F$.

In particular,
\begin{equation}    \label{MeanDzetaPolynomial}
\sum\limits _{i=1} ^h t^{g-g_i} P_{C_i}(t) = L_F(t)
\end{equation}
for the $\zeta$-polynomials $P_{C_i}(t)$ of $C_i = \mathcal{E}_D \mathcal{L} (G_i)$ and the Hasse-Weil polynomial $L_F(T)$ of the function field $F$.
 \end{proposition}

  \begin{proof}

   Note that (\ref{RelationDzetaPolynomialsAGCodes}) is an equality of homogeneous polynomials of $x$ and $y$ of degree $n$, whose monomials are of degree $s \geq 1$ with respect to $y$.
 Therefore (\ref{RelationDzetaPolynomialsAGCodes}) is equivalent to
 \begin{equation}   \label{CoeffLFSumWC}
 \begin{split}
 \frac{\sum\limits _{i=1} ^{h(F)} \mathcal{W}_{C_i} ^{(s)}}{q-1} = {\rm Coeff} _{x^{n-s} y^{s} t^{m}} ( \zeta _F(t) [ y (1-t) + xt] ^n ) =  \\
 {\rm Coeff} _{t^{m}} \left( \binom{n}{s} t^{n-s} (1-t)^s \zeta _F(t) \right) =  \\
 \binom{n}{s} {\rm Coeff} _{t^{s-n+m}} \left( (1-t) ^s \zeta _F(t) \right)
 \end{split}
  \end{equation}
 for $\forall s \in {\mathbb N}$.
 Note that $C_i$ are of minimum distance $d(C_i) \geq n-m$, so that $\mathcal{W}_{C_i} ^{(s)} =0$ for $1 \leq s < n-m$.
 On the other hand,
 \[
 (1-t) ^s \zeta _F(t) = \frac{(1-t)^{s-1} L_F(t)}{1-qt}
 \]
 has no pole at $t=0$, so that ${\rm Coeff} _{t^{s-n+m}} \left( (1-t)^s \zeta _F(t) \right) =0$ for $s-n+m <0$, $s \in {\mathbb N}$.
 That is why it suffices to verify (\ref{CoeffLFSumWC}) for $s \geq n-m$, $s \in {\mathbb N}$.

 Note that the number of the codewords $c = (f(P_1), \ldots , f(P_n)) \in C_i$, $f \in L(G_i)$ of weight $s$ equals the number of the rational functions
  $f \in L(G_i) \setminus \{ 0 \}$, vanishing at $n-s$ of the points $P_1, \ldots , P_n$.
Bearing in mind that the projective space ${\mathbb P} (L(G_i)) = {\mathbb P}^{m-g} ( {\mathbb F}_q)$ parameterizes the effective divisors, linearly equivalent to $G_i$ and two rational functions $f, f'\in F \setminus \{ 0 \}$ have  one and a same divisor exactly when they are on one and a same ${\mathbb F}_q^*$-orbit, $f'\in {\mathbb F}_q^* f$, one concludes that $\frac{\mathcal{W}_{C_i}^{(s)}}{q-1}$ is the number of the effective divisors $E = (f) + G_i$, which are linearly equivalent to $G_i$ with $| {\rm Supp} (E) \cap {\rm Supp} (D) | = n-s$.
Thus,
\[
e_{m,s} := \frac{\sum\limits _{i=1} ^{h(F)} \mathcal{W}_{C_i} ^{(s)}}{q-1}
\]
equals the number of the effective divisors $E \in {\rm Div} (F) ^{\geq 0}$ of degree $\deg E = m$ with $|{\rm Supp} (E) \cap {\rm Supp} (D)| = n-s$.
For any $s$-tuple of indices $i = \{ i_1, \ldots , i_s \}$, $1 \leq i_1 < \ldots < i_s \leq n$ let $D_i := P_{i_1} + \ldots + P_{i_s}$ and $e_m (i)$ be the number of the effective divisors $E \in {\rm Div} (F) ^{\geq 0}$ of degree $\deg E =m$ with ${\rm Supp} (E) \cap {\rm Supp} (D) = {\rm Supp} (D - D_i)$.
Then $e_{m,s} = \sum\limits _{i} e_m (i)$ and it suffices to show that $ e_m (i) = {\rm Coeff} _{t^{s-n+m}} \left( (1-t) ^s \zeta _F(t) \right)$ for any  $i$, in order to justify (\ref{CoeffLFSumWC}) and (\ref{RelationDzetaPolynomialsAGCodes}).

 To this end, observe that $E \in {\rm Div} (F) ^{\geq 0}$ is an effective divisor of degree $\deg E =m$ with
 ${\rm Supp} (E) \cap {\rm Supp} (D) = {\rm Supp} (D-D_i)$ if and only if the difference $E_1 := E - (D - D_i) \in {\rm Div} (F) ^{\geq 0}$ is an effective divisor of degree $\deg E_1 = m-n+s$ with support  ${\rm Supp} (E_1) \cap {\rm Supp} (D_i) = \emptyset$.
Now, $e_m (i)$ equals the number of the effective divisors $E_1 \in {\rm Div} (F) ^{\geq 0}$ of degree $\deg E_1 = m-n+s$ with
 ${\rm Supp} (E_1) \cap {\rm Supp} (D_i) = \emptyset$.
Recall that the Hasse-Weil $\zeta$-function
\[
\zeta _F(t) = \prod\limits _{\nu \in \mathcal{P}} \frac{1}{1 - t^{\deg \nu}} = \sum\limits _{i=0} ^{\infty} \mathcal{A}_i t^{i}
\]
is the generating function for the number $\mathcal{A}_i$ of the effective divisors of $F$ of degree $i$.
Bearing in mind that $D_i = \nu _{i_1} + \ldots + \nu _{i_s}$ is a sum of $s$ different places $i_j$ of degree $\deg \nu _{i_j} =1$, one observes that
$(1-t) ^s \zeta _F(t)$ is the generating function for the number of the effective divisors of $F$ of degree $i$, whose support is disjoint with
${\rm Supp} (D_i)$.
In other words, $e_m (i) = {\rm Coeff} _{t^{m-n+s}} \left( (1-t)^s \zeta _F(t) \right)$.

The equality (\ref{MeanDzetaPolynomial}) is an immediate consequence of Proposition \ref{LinearCode},  (\ref{RelationDzetaPolynomialsAGCodes}) and the fact that $C_i = \mathcal{E}_D \mathcal{L} (G_i)$ are of dimension $\dim _{{\mathbb F}_q} C_i = l(G_i) = m-g+1$,  minimum distance  $d_i \geq n-m$ and, therefore, of genus \[
g_i = n+1 - \dim _{{\mathbb F}_q} C_i - d_i = n-m - d_i + g \leq  g.
\]

\end{proof}

Proposition \ref{GeometricDzetaPolynomialAGCode} motivates Duursma to refer to $P_C(t)$ as to the zeta polynomial of an arbitrary linear code
$C \subset {\mathbb F}_q^n$.
He establishes that  $P_C(t)$ and $\mathcal{W}_C(x,y)$  are in a bijective correspondence and Mac Williams identities, relating the weight distributions
 $\{ \mathcal{W}_C^{(w)} \} _{w=d} ^n$, $\{ \mathcal{W}_{C^{\perp}} ^{(w)} \} _{w= d^{\perp}} ^n$ of a pair $(C, C^{\perp})$ of mutually dual linear codes are equivalent to the functional equation
\begin{equation}     \label{DuursmaFunctionalEquation}
P_{C^{\perp}} (t) = P_C \left( \frac{1}{qt} \right) q^g t^{g + g^{\perp}}
\end{equation}
 for the corresponding zeta polynomials $P_C(t), P_{C^{\perp}} (t)$.

 In \cite{D1} and \cite{D3} Duursma observes the existence of a polynomial $D_C(t) = \sum\limits _{i=0} ^{r-2} c_i t^{i} \in {\mathbb Q}[t]$, defined by the identity
 \[
 P_C(t) = (1-t)(1-qt) D_C(t) + t^g
\]
of polynomials in  $t$, but does not make use of $D_C(t)$ for the study of the homogeneous weight enumerator $\mathcal{W}_C(x,y)$ of $C$.
  He mentions in \cite{D3}  that the analogue $D_F(t)$ of $D_C(t)$ for a function field $F$ of one variable accounts for the contribution of the special divisors of $F$ to the zeta function $Z_F(t)$.
  From now on, we refer to $D_C(t)$ as to Duursma's reduced polynomial of $C$.

   The present note provides an explicit bijective correspondence between the weight distribution $\{ \mathcal{W}_C^{(w)} \} _{w=d} ^{n}$ of an arbitrary linear code $C \subset {\mathbb F}_q^n$  and the coefficients $\{ c_i \} _{i=0} ^{r-2}$ of its Duursma's reduced polynomial
    $D_C(t) = \sum\limits _{i=0} ^{r-2} c_i t^{i}$  (cf. Proposition \ref{DExpressionOfW}).

The classical Hasse-Weil Theorem establishes that all the roots of the Hasse-Weil polynomial $L_F(t) \in {\mathbb Z}[t]$ of the function field
$F = {\mathbb F}_q(X)$ of a curve $X$ of genus $g$ over ${\mathbb F}_q$ are on the circle
$S \left( \frac{1}{\sqrt{q}} \right) : \left \{ z \in {\mathbb C} \ \ \Big \vert \ \ |z| = \frac{1}{\sqrt{q}} \right \}$
(cf. Theorem 4.2.3 form \cite{NX}).
Duursma says that a linear code $C \subset {\mathbb F}_q^n$ satisfies the Riemann Hypothesis Analogue  if all the roots of its zeta polynomial $P_C(t) = \sum\limits _{i=0} ^r a_i t^{i} \in {\mathbb Q}[t]$ are on the circle $S \left( \frac{1}{\sqrt{q}} \right)$.
Let $C$ be an ${\mathbb F}_q$-linear code of dimension $k$ and minimum distance $d$, which satisfies the Riemann Hypothesis Analogue.
    Proposition \ref{RHAImpliesFSD} shows that  $C$ is formally self-dual, while   Corollary \ref{BoundOnFieldCardinality} provides an  explicit upper bound on  the cardinality $q$ of the basic field, depending on $k$ and $d$.
 Let us recall that $C$ is formally self-dual if it has the same weight distribution $\mathcal{W}_C ^{(w)} = \mathcal{W}_{C^{\perp}} ^{(w)}$, $\forall 0 \leq w \leq n$ as its dual code $C^{\perp} \subset {\mathbb F}_q^n$.
In the light of Duursma's results and our Proposition \ref{DExpressionOfW}, the formal self-duality of $C$ turns to be equivalent to the functional equation
$
P_C(t) = P_C \left( \frac{1}{qt} \right) q^g t^{2g}
$
for $P_C(t)$ and to the functional equation
$
D_C(t) = D_C \left( \frac{1}{qt} \right) q^{g-1} t^{2g-2}
$
for $D_C(t)$.
 Proposition \ref{FSD} from the present note  expresses explicitly the homogeneous weight enumerator $\mathcal{W}_C(x,y)$ of a formally self-dual code
 $C \subset {\mathbb F}_q^n$ by the lowest  half of the coefficients of $D_C(t)$ or by the numbers $\mathcal{W}_C^{(d)}, \ldots , \mathcal{W}_C^{(k)}$ of the codewords $c \in C$, whose weights  are between the minimum distance $d$ of $C$ and the dimension $k$.

In  \cite{DL} Dodunekov and Landgev introduce the near-MDS code $C \subset {\mathbb F}_q^n$ as the ones with quadratic zeta polynomial $P_C(t)$.
Kim and Hyun's article \cite{KH} provides a necessary and sufficient condition for a near-MDS code to satisfy the Riemann Hypothesis Analogue.
Note that the zeta polynomials $P_C(t)$ and Duursma's reduced polynomials $D_C(t)$  of formally self-dual codes $C \subset {\mathbb F}_q^n$ are of even degree.
Our Proposition  \ref{RHAForQuadraticD}  is a necessary and sufficient condition for a formally self-dual code $C \subset {\mathbb F}_q^n$   with zeta polynomial $P_C(T)$ of   $\deg P_C(t) = 4$ to be subject to the Riemann Hypothesis Analogue.
Let $S_{\nu}$, $\nu \in {\mathbb N}$ be the uniquely determined logarithmic coefficients of $P_{\mathcal{C}}(t)$, defined by the equality of formal power series
$\log P_{\mathcal{C}}(t) = \sum\limits _{\nu =1} ^{\infty} S_{\nu} \frac{t^{\nu}}{\nu} \in {\mathbb C} [[t]]$.
Adapting Bombieri's proof of the Hasse-Weil Theorem, \cite{KMT} shows that a linear code $\mathcal{C}$ satisfies the Riemann Hypothesis Analogue exactly when the sequence $\{ S_{\nu} q^{ - \frac{\nu}{2}} \} _{\nu =1} ^{\infty} \subset {\mathbb C}$ is absolutely bounded.

The last, third section is devoted to Duursma's reduced polynomial $D_F(t)$ of the function field $F = {\mathbb F}_q(X)$ of a curve
 $X / {\mathbb F}_q \subset {\mathbb P}^N ( \overline{{\mathbb F}_q})$ of genus $g$ over ${\mathbb F}_q$.
It establishes that $D_F(t) \in {\mathbb Z}[t]$ is determined uniquely by its lowest $g$ coefficients, which equal the numbers $\mathcal{A}_i$ of the effective divisors of $F$ of degree $0 \leq i \leq g-1$.
Our Proposition  \ref{HasseWeilDecomposition} shows that the zeta function
\[
\frac{D_F(t)}{(1-t)(1-qt)} = \sum\limits _{i=0} ^{\infty} \mathcal{B}_i t^{i},
\]
associated with $D_F(t)$ has the properties of a generating function for the numbers $\mathcal{B}_i$ of the effective divisors of degree $i \geq 0$ of a virtual function field of genus $g-1$ over ${\mathbb F}_q$.
There arises the following

{\bf Open Problem:}  To characterize the function fields $F = {\mathbb F}_q(X)$   of curves $X / {\mathbb F}_q \subset {\mathbb P}^N ( \overline{{\mathbb F}_q})$ of genus $g$ over ${\mathbb F}_q$, for which there are curves $Y / {\mathbb F}_q \subset {\mathbb P}^M ( \overline{{\mathbb F}_q})$ of genus $g-1$, defined over ${\mathbb F}_q$ with Hasse-Weil zeta function
\[
Z_{{\mathbb F}_q(Y)} (t) = \frac{D_F(t)}{(1-t)(1-qt)}.
\]

\section{The homogeneous weight enumerator of an arbitrary code}

\begin{proposition}      \label{DExpressionOfW}
Let $C \subset {\mathbb F}_q^n $ be a  linear code of dimension $k = \dim _{{\mathbb F}_q} C$, minimum distance $d$ and genus
 $g = n+1 -k-d \geq 1$, whose dual $C^{\perp} \subset {\mathbb F}_q^n$ is of minimum distance $d^{\perp}$ and genus
$g^{\perp} = k +1 - d^{\perp} \geq 1$.
If
\[
D_C(t) = \sum\limits _{i=0} ^{g + g^{\perp} -2} c_i t^{i} \in {\mathbb Q}[t]
 \]
 is Duursma's reduced polynomial of $C$ and $\mathcal{M}_{n, n+1 -k} (x,y)$ is the homogeneous weight enumerator of  an MDS-code of length $n$, dimension $k$  and minimum distance $ n+1 -k$, then the homogeneous weight enumerator of $C$ is
\begin{equation}    \label{WByD}
\mathcal{W}_C(x,y) = \mathcal{M}_{n, n+1-k} (x,y)  + (q-1) \sum\limits _{i=0} ^{g + g^{\perp} -2} c_i \binom{n}{d+i} (x-y)^{n-d-i} y^{d+i}.
\end{equation}
More precisely, Duursma's reduced polynomial $D_C(t) = \sum\limits _{i=0} ^{g + g^{\perp} -2} c_i t^{i}$ determines uni\-que\-ly the weight distribution of $C$, according to
\begin{equation}    \label{WeightsByD1}
\mathcal{W}_C^{(w)} = (q-1)  \binom{n}{w} \sum\limits _{i=0} ^{w-d} (-1)^{w-d-i} \binom{w}{d+i} c_i \ \ \mbox{   for  } \ \ d \leq w \leq d + g -1,
\end{equation}
\begin{equation}    \label{WeightsByD2}
\begin{split}
\mathcal{W}_C^{(w)} = & (q-1) \binom{n}{w} \sum\limits _{i=0} ^{\min (w-d, n-d-d^{\perp})} (-1)^{w-d-i} \binom{w}{d+i} c_i     \\
& + \binom{n}{w} \sum\limits _{j=0} ^{w-n-1+k} (-1)^{j} \binom{w}{j} (q^{w-n+k-j} -1) \ \ \mbox{   for  } \ \  d+g \leq w \leq n.
\end{split}
\end{equation}
Conversely, for $\forall 0 \leq i \leq g + g^{\perp} -2$ the numbers $\mathcal{W}_C ^{(d)}, \ldots , \mathcal{W}_C^{(d+i)}$ determine uniquely the coefficient $c_i$ of Duursma's reduced polynomial $D_C(t) = \sum\limits _{i=0} ^{g + g^{\perp} -2} c_i t^{i}$  by
\begin{equation}    \label{DCByW1}
c_i = (q-1)^{-1} \binom{n}{d+i} ^{-1} \sum\limits _{w=d} ^{d+i} \binom{n-w}{n -d-i} \mathcal{W}_C ^{(w)}
\end{equation}
for $0 \leq i \leq g-1$,
\begin{equation}    \label{DCByW2}
\begin{split}
c_i = (q-1)^{-1} \binom{n}{d+i} ^{-1} \left \{  \sum\limits _{w=d} ^{d+g-1} \binom{n-w}{n-d-i} \mathcal{W}_C ^{(w)}  \right. \\
 \left.  + \sum\limits _{w=d+g} ^{d+i} \binom{n-w}{n-d-i}
   \left[ \mathcal{W}_C ^{(w)} - \binom{n}{w} \sum\limits _{j=0} ^{w-n-1+k} (-1)^{j} \binom{w}{j} (q^{w-n+k-j} -1) \right] \right \}
\end{split}
\end{equation}
for $g \leq i \leq g + g^{\perp} -2$.

In particular,
\[
(q-1) \binom{n}{d+i} c_i \in {\mathbb Z}
\]
are integers for all $0 \leq i \leq g + g^{\perp} -2$.

 The aforementioned formulae imply that  $\mathcal{W}_C^{(d)}, \ldots , \mathcal{W}_C^{(d + g + g^{\perp} -2)}$ determine uniquely the homogeneous weight enu\-me\-ra\-tor $\mathcal{W}_C(x,y)$ of $C$ by the formula
\begin{equation}    \label{HWEByFirstCoeff}
\mathcal{W}_C(x,y) = \sum\limits _{w=d} ^{d + g + g^{\perp} -2} \mathcal{W}_C ^{(w)} \lambda  _w(x,y) + \Lambda (x,y),
\end{equation}
with explicit polynomials
\begin{equation}    \label{PhiFactor}
\lambda _w (x,y) := \sum\limits _{s=w} ^{d + g + g^{\perp} -2} \binom{n-w}{n-s} (x-y) ^{n-s} y^s  \ \ \mbox{  for } \ \ d \leq w \leq d + g + g^{\perp} -2
\end{equation}
and
\begin{equation}     \label{ExplicitTerm}
\Lambda (x,y) := \mathcal{M}_{n, n+1-k} (x,y) - \sum\limits _{w=d+g} ^{d+g + g^{\perp} -2} \mathcal{M}_{n, n+1-k} ^{(w)} \lambda _w (x,y).
\end{equation}
\end{proposition}

\begin{proof}

In the case of $g=0$, note that $C$ is an MDS-code and $\mathcal{W}_C (x,y) = \mathcal{M}_{n,n+1-k}(x,y)$.
Form now on, we assume that $g >0$ and put $r:= g + g^{\perp}$.
Making use of $d + g = n+1-k$, let us express
\[
\mathcal{W}_C(x,y) = \mathcal{M}_{n, d+g} (x,y) + \sum\limits _{i=0} ^{r} b_i \mathcal{M}_{n, d+i} (x,y)
\]
by  some rational numbers $b_i \in {\mathbb Q}$.
Then the seta polynomial $P_C(t) = t^{g} + \sum\limits _{i=0} ^r b_i t^{i}$    and
Duursma's reduced polynomial $D_C(t) = \sum\limits _{i=0} ^{r-2} c_i t^{i}$ of $C$ are related by  the equality
\begin{equation}    \label{DefD}
P_C(t) - t^g = (1-t)(1-qt) D_C(t).
\end{equation}
Let us introduce $c_{-2} = c_{-1} = c_{r-1} = c_r =0$ and compare the coefficients of $t^{i}$   from the left and right hand side of (\ref{DefD}), in order to obtain
\[
b_i = c_i -(q+1) c_{i-1} + qc_{i-2} \ \ \mbox{ for } \ \ \forall  0 \leq i \leq r.
\]
Therefore
 \begin{align*}
\mathcal{W}_C(x,y)   =  \mathcal{M}_{n, d+g} (x,y) + \sum\limits _{i=0} ^r c_i \mathcal{M}_{n, d+i} (x,y) \\
  - (q+1) \sum\limits _{i=0} ^r c_{i-1} \mathcal{M}_{n, d+i} (x,y) + q \sum\limits _{i=0} ^r c_{i-2} \mathcal{M}_{n, d+i} (x,y).
\end{align*}
Setting $j = i-1$, respectively, $j = i-2$ in the last two sums, one obtains
\begin{align*}
\mathcal{W}_C(x,y) = \mathcal{M}_{n, d+g} (x,y) + \sum\limits _{i=0} ^r c_i \mathcal{M}_{n, d+i} (x,y)  \\
 - (q+1) \sum\limits _{j=-1} ^{r-1} c_j \mathcal{M}_{n, d+j+1} (x,y) +q \sum\limits _{j=-2} ^{r-2} c_j \mathcal{M}_{n, d+j +2} (x,y),
 \end{align*}
whereas
\begin{equation}   \label{TheEqualityUnder Study}
\begin{split}
\mathcal{W}_C (x,y) = \mathcal{M}_{n, d+g} (x,y) \\
 + \sum\limits _{j=0} ^{r-2} c_j [ \mathcal{M}_{n, d+j} (x,y) -   (q+1) \mathcal{M}_{n, d+j+1} (x,y) + q \mathcal{M}_{n, d+j+2} (x,y) ].
\end{split}
\end{equation}
Let us put
\[
\mathcal{W}_{n, d+j} (x,y) := \mathcal{M}_{n, d+j} (x,y) - (q+1) \mathcal{M}_{n, d+j+1} (x,y) + q \mathcal{M}_{n, d+j+2} (x,y)
\]
and recall that the homogeneous weight enumerator of an MDS-code of length $n$ and minimum distance $d+j$  is
\[
\mathcal{M}_{n, d+j} (x,y) = x^n + \sum\limits _{w = d+j} ^n \mathcal{M}_{n, d+j}  ^{(w)} x^{n-w} y^w
\]
with
\begin{equation}   \label{MDSWeightDistribution}
\mathcal{M}_{n, d+j} ^{(w)} = \binom{n}{w} \sum\limits _{i=0} ^{w - d - j} (-1)^{i} \binom{w}{i} ( q^{w+1 -d-j-i} -1).
\end{equation}
Therefore
  \begin{align*}
 \mathcal{W}_{n, d+j} (x,y) = \mathcal{M}_{n, d+j} ^{(d+j)} x^{n -d-j} y^{d+j} +
 [ \mathcal{M}_{n, d+j} ^{(d+j+1)} - (q+1) \mathcal{M}_{n, d+j+1} ^{(d+j+1)} ] x^{n-d-j-1} y^{d+j+1} \\
 + \sum\limits _{w=d+j+2} ^n [ \mathcal{M}_{n, d+j} ^{(w)} - (q+1) \mathcal{M}_{n, d+j+1} ^{(w)} + q \mathcal{M}_{n, d+j +2} ^{(w)} ] x^{n-w} y^w.
 \end{align*}
Making use of the weight distribution (\ref{MDSWeightDistribution}) of an MDS-code and introducing
\[
\mathcal{W}_{n, d+j} ^{(w)} := \mathcal{M}_{n, d+j} ^{(w)} - (q+1) \mathcal{M}_{n, d+j+1} ^{(w)} + q \mathcal{M}_{n,d+j+2} ^{(w)} \ \ \mbox{  for  } \ \
d+j+2 \leq w \leq n,
\]
one expresses
\begin{align*}
\mathcal{W}_{n, d+j} (x,y) = \binom{n}{d+j} (q-1) x^{n-d-j} y^{d+j} \\
 - \binom{n}{d+j+1} (q-1) (d+j+1) x^{n-d-j-1} y^{d+j+1} + \sum\limits _{w= d+j +2} ^n \mathcal{W}_{n, d+j} ^{(w)} x^{n-w} y^w.
\end{align*}
For any $d + j +2 \leq w \leq n$ one has
 \[
\mathcal{W}_{n, d+j} ^{(w)} = \binom{n}{w} \binom{w}{d+j} (q-1) (-1) ^{w-d-j}.
\]
Making use of
\[
\binom{n}{w} \binom{w}{d+j} = \binom{n-d-j}{w-d-j} \binom{n}{d+j},
\]
one obtains
\begin{align*}
\mathcal{W}_{n, d+j} (x,y) =  \\
 \binom{n}{d+j} (q-1) x^{n-d-j} y^{d+j} -
\binom{n}{d+j+1} (q-1) (d+j+1) x^{n-d-j-1} y^{d+j+1} +  \\
+ \sum\limits _{w=d+j+2} ^n \binom{n}{d+j} \binom{n-d-j}{w-d-j} (q-1) (-1)^{w-d-j} x^{n-w} y^w.
\end{align*}
 Bearing in mind that
\[
(d+j+1) \binom{n}{d+j+1} = (n-d-j) \binom{n}{d+j},
\]
one derives that
\begin{align*}
\mathcal{W}_{n, d+j} (x,y) = \binom{n}{d+j} (q-1) \left[ x^{n-d-j} y^{d+j} - (n-d-j) x^{n-d-j-1} y^{d+j+1} + \right.  \\
\left. +\sum\limits _{w = d+j+2} ^n (-1)^{w-d-j} \binom{n-d-j}{w-d-j} x^{n-w} y^w \right].
\end{align*}
Introducing $s:= w - d - j$, one expresses
\[
\sum\limits _{w=d+j+2} ^n (-1) ^{w-d-j} \binom{n-d-j}{w-d-j} x^{n-w} y^w =
\sum\limits _{s=2} ^{n-d-j} (-1)^{s} \binom{n-d-j}{s} x^{n-d-j-s} y^{d+j+s}
\]
and concludes that
\begin{equation}   \label{TheCore}
\mathcal{W}_{n, d+j} (x,y) = \binom{n}{d+j} (q-1) (x-y) ^{n-d-j} y^{d+j}.
\end{equation}
The equality $\mathcal{W}_{n, n-k} (x,y) = \binom{n}{k} (q-1) (x-y)^k y^{n-k}$ is exactly the claim (c) of Lemma 1 from Kim and Nyun's work \cite{KH}.
Plugging in  (\ref{TheCore}) in (\ref{TheEqualityUnder Study}) and bearing in mind that $d +g = n+1 -k$, one obtains (\ref{WByD}).

In order to prove (\ref{WeightsByD1}) and (\ref{WeightsByD2}), let us put
\[
\mathcal{V}_C (x,y) := \mathcal{W}_C (x,y) - \mathcal{M}_{n, n+1-k} (x,y)
\]
and note that $\mathcal{V}_C (x,y) = \sum\limits _{w=d} ^n \mathcal{V}_C ^{(w)} x^{n-w} y^w$ with  $\mathcal{V}_C ^{(w)} = \mathcal{W}_C ^{(w)}$
for $d \leq w \leq n-k,$
\[
\mathcal{V}_C ^{(w)} = \mathcal{W}_C^{(w)} - \mathcal{M}_{n, n+1-k} ^{(w)}  = \mathcal{W}_C ^{(w)} - \binom{n}{w} \sum\limits _{i=0} ^{w-n-1+k} (-1)^{i} \binom{w}{i} (q^{w-n+k-i} -1)
\]
for $d + g = n+1 -k \leq w \leq n$.
Making use of (\ref{WByD}), one expresses
\begin{align*}
\mathcal{V}_C (x,y) = (q-1) \sum\limits _{i=0} ^{g + g^{\perp}-2} c_i \binom{n}{d+i} \sum\limits _{s=0} ^{n-d-i} \binom{n-d-i}{s} (-1)^{n-d-i-s} x^s y^{n-s}  \\
= (q-1) \sum\limits _{s=0} ^{n-d} \left[ \sum\limits _{i=0} ^{\min (n-d-s, g + g^{\perp} -2)} c_i \binom{n}{d+i} \binom{n-d-i}{s} (-1)^{n-d-i-s} \right] x^s y^{n-s},
\end{align*}
after changing the summation order.
Setting  $w:= n-s$, one obtains
\[
\mathcal{V}_C (x,y) = (q-1) \sum\limits _{w=d} ^n \left[ \sum\limits _{i=0} ^{\min (w-d, n-d-d^{\perp})} c_i \binom{n}{d+i} \binom{n-d-i}{n-w} (-1)^{w-d-i} \right] x^{n-w} y^w.
\]
Then
\[
\binom{n}{d+i} \binom{n-d-i}{n-w} = \binom{n}{w} \binom{w}{d+i},
\]
allows to  concludes that
\[
\mathcal{V}_C ^{(w)} = (q-1) \binom{n}{w} \sum\limits _{i=0} ^{\min (w-d, n-d-d^{\perp})} c_i \binom{w}{d+i} (-1) ^{w-d-i}
 \ \ \mbox{   for  } \ \ \forall d \leq w \leq n,
\]
which proves (\ref{WeightsByD1}), (\ref{WeightsByD2}).

Towards (\ref{DCByW1}), (\ref{DCByW2}), let us introduce $z := x-y$ and express (\ref{WByD}) in the form
\begin{equation}    \label{WorkDByW}
\mathcal{V}_C (y+z,y) = (q-1) \sum\limits _{i=0} ^{g + g^{\perp} -2} c_i \binom{n}{d+i} z^{n-d-i} y^{d+i}.
\end{equation}
On the other hand,
\begin{align*}
\mathcal{V}_C (y+z,y) = \sum\limits _{w=d} ^n \mathcal{V}_C ^{(w)} (y+z) ^{n-w} y^w \\
= \sum\limits _{w=d} ^n \sum\limits _{s=0} ^{n-w} \binom{n-w}{s} \mathcal{V}_C ^{(w)} y^{n-s} z^s =
\sum\limits _{s=0} ^{n-d} \left[ \sum\limits _{w=d} ^{n-s} \binom{n-w}{s} \mathcal{V}_C ^{(w)} \right] y^{n-s} z^s,
\end{align*}
after changing the summation order.
Comparing the coefficients of $y^{d+i}z^{n-d-i}$ in the left and right hand side of (\ref{WorkDByW}), one obtains
\[
\sum\limits _{w=d} ^{d+i} \binom{n-w}{n-d-i} \mathcal{V}_C ^{(w)} = (q-1) c_i \binom{n}{d+i},
\]
whereas
\[
c_i = (q-1)^{-1} \binom{n}{d+i}^{-1} \sum\limits _{w=d} ^{d+i} \binom{n-w}{n-d-i} \mathcal{V}_C ^{(w)}.
\]
Combining with (\ref{MDSWeightDistribution}), one justifies (\ref{DCByW1}) and  (\ref{DCByW2}).
These formulae imply also the fact  that $(q-1) \binom{n}{d+i} c_i \in {\mathbb Z}$ are integers for all $0 \leq i \leq g + g^{\perp} -2$.

The substitution by  (\ref{DCByW1}), (\ref{DCByW2}),  (\ref{MDSWeightDistribution}) in  (\ref{WByD}) yields
\begin{align*}
\mathcal{W}_C (x,y)= \mathcal{M}_{n,n+1-k} (x,y)
 + \sum\limits _{i=0} ^{g + g^{\perp} -2} \sum\limits _{w=d} ^{d+i} \binom{n-w}{n-d-i} \mathcal{W}_C ^{(w)} (x-y)^{n-d-i} y^{d+i} \\
  - \sum\limits _{i=g} ^{g + g^{\perp} -2} \sum\limits _{w=d+g} ^{d+i} \binom{n-w}{n-d-i} \mathcal{M}_{n,n+1-k} ^{(w)} (x-y) ^{n-d-i} y^{d+i}.
  \end{align*}
  One exchanges the summation order in the double sums towards
  \begin{align*}
  \mathcal{W}_C  (x,y) = \mathcal{M}_{n,n+1-k} (x,y)
  + \sum\limits _{w=d} ^{d + g + g^{\perp} -2} \mathcal{W}_C ^{(w)} \sum\limits  _{i=w-d} ^{g + g^{\perp} -2} \binom{n-w}{n-d-i} (x-y) ^{n-d-i} y^{d+i}  \\
  - \sum\limits _{w=d+g} ^{d + g + g^{\perp} -2} \mathcal{M}_{n,n+1-k} ^{(w)} \sum\limits _{i=w-d} ^{g + g^{\perp} -2} \binom{n-w}{n-d-i}
   (x-y) ^{n-d-i} y^{d+i}.
  \end{align*}
  Introducing $s := d+i$, one obtains (\ref{HWEByFirstCoeff}) with (\ref{PhiFactor}) and (\ref{ExplicitTerm}).

\end{proof}

Comparing the coefficients of $x^{n-d} y^d$
 in the left and right hand sides of (\ref{WByD}), one obtains
 $\mathcal{W}_C ^{(d)} = (q-1) \binom{n}{d} c_0$ for a linear code $C$ of genus $g  \geq 1$.
We claim that  $c_0 <1$.
To this end, note that for any $d$-tuple  $\{ i_1, \ldots , i_d \} \subset \{ 1, \ldots , n \}$, supporting a word $c \in C$ of weight $d$ there are exactly $q-1$ words $c'\in C$ with ${\rm Supp} (c')  = {\rm Supp} (c) = \{ i_1, \ldots , i_d \}$.
That is due to the fact that the columns $H_{i_1}, \ldots , H_{i_d}$ of an arbitrary parity check matrix $H$ of $C$ are of rank $d-1$ and there are no words of weight $\leq d-1$ in the right null space of the matrix $(H_{i_1} \ldots H_{i_d})$.
It is clear that $\nu \leq \binom{n}{d}$, so that
\[
c_0 = \frac{\nu}{\binom{n}{d}} \leq 1.
\]
If we assume that $c_0 =1$ then any $d$-tuple of columns of $H$ is linearly dependent.
Bearing in mind that ${\rm rk} H = n-k$, one concludes that $d > n-k$.
Combining with Singleton Bound $d \leq n-k+1$, one obtains $d = n-k+1$.
That contradicts the assumption that $C$ is not an MDS-code and proves that $c_0 <1$ for any ${\mathbb F}_q$-linear code $C \subset {\mathbb F}_q^n$ of genus
 $g \geq 1$.
Note that $c_0$ can be interpreted as the probability for a $d$-tuple to support a word of weight $d$ from $C$.

\section{The Riemann Hypothesis Analogue and the formal self-duality of a linear code}

Recall that a linear code $C \subset {\mathbb F}_q^n$ with dual code $C^{\perp} \subset {\mathbb F}_q^n$ is formally self-dual  if $C$ and $C^{\perp}$ have one and a same number  $\mathcal{W}_C ^{(w)} = \mathcal{W}_{C^{\perp}} ^{(w)}$  of codewords of weight $0 \leq w \leq n$.
Let us mention some trivial consequences of the formal self-duality of $C$.
First of all, $C$ and $C^{\perp}$ have one and a same minimum distance $d = d(C) = d(C^{\perp}) = d^{\perp}$.
Further, $C$ and $C^{\perp}$ have one and a same cardinality
\[
q^{\dim C} = \sum\limits _{w=0} ^n \mathcal{W}_C ^{(w)} = \sum\limits _{w=0} ^n \mathcal{W}_C ^{(w)} = q^{\dim C^{\perp}},
\]
so that $k = \dim C = \dim C^{\perp} = k^{\perp}$ and the length $n = k + k^{\perp} = 2k$ is an even integer.
The genera $g  = k+1 -d = g ^{\perp}$ also coincide.
Let $P_C(t) = \sum\limits _{i=0} ^{2g} a_i t^{i}$  and $P_{C^{\perp} } = \sum\limits _{i=0} ^{2g} a_i ^{\perp} t^{i}$ be the zeta polynomials of $C$, respectively, of $C^{\perp}$.
The consecutive comparison of the coefficients of $x^{n-d}y^d, x^{n-d-1}y^{d+1}, \ldots , x^{n-d-2g} y^{d+2g}$ from the homogeneous polynomial
\begin{align*}
a_0 \mathcal{M}_{2k,d} (x,y) + a_1 \mathcal{M}_{2k,d+1}(x,y) + \ldots + a_{2g} \mathcal{M}_{2k,d+2g}(x,y) = \mathcal{W}_C(x,y)  \\
= \mathcal{W}_{C^{\perp}} (x,y) =
 a_0^{\perp} \mathcal{M}_{2k,d} (x,y) + a_1^{\perp} \mathcal{M}_{2k,d+1}(x,y) + \ldots + a_{2g}^{\perp} \mathcal{M}_{2k,d+2g}(x,y)
\end{align*}
in $x,y$ yields $a_i = a_i ^{\perp}$ for $\forall 0 \leq i \leq 2g$.
It is clear that $a_i = a_i ^{\perp}$ for $\forall 0 \leq i \leq 2g$ suffices for $\mathcal{W}_C(x,y) = \mathcal{W}_{C^{\perp}} (x,y)$, so that the formal self-duality of $C$ is tantamount to the coincidence $P_C(t) = P_{C^{\perp}}(t)$ of the zeta polynomials of $C$ and $C^{\perp}$.
Duursma has shown that Mac Williams identities for $\mathcal{W}_C ^{(w)}$ and $\mathcal{W}_{C^{\perp}}^{(w)}$ are equivalent to the functional equation
(\ref{DuursmaFunctionalEquation})  for the zeta polynomials $P_C(t)$, $P_{C^{\perp}} (t)$ of $C, C^{\perp} \subset {\mathbb F}_q^n$ with genera $g, g^{\perp}$.
Thus, an ${\mathbb F}_q$-linear code $C \subset {\mathbb F}_q^n$ is formally self-dual if and only if its zeta polynomial $P_C(t)$ satisfies the functional equation
\begin{equation}    \label{FuncEqFuncFieldGenusGCurveOverFQ}
P_C(t) = P_C \left( \frac{1}{qt} \right) q^g t^{2g}
\end{equation}
 of the Hasse-Weil polynomial of the  function field of a curve of genus $g$ over ${\mathbb F}_q$.

\begin{proposition}     \label{RHAImpliesFSD}
If a linear code $C \subset {\mathbb F}_q^n$ satisfies the Riemann Hypothesis Analogue then  $C$ is formally self-dual, i.e.,
 the zeta polynomial $P_C(t)$ of $C$ is subject to the functional equation  (\ref{FuncEqFuncFieldGenusGCurveOverFQ})
of the Hasse-Weil polynomial of the function field of a curve of genus $g$ over ${\mathbb F}_q$.
\end{proposition}

\begin{proof}

Let us assume that $P_C(t)$ of degree $r :=  g + g^{\perp}$ satisfies the Riemann Hypothesis Analogue, i.e.,
\[
P_C(t) = a_r \prod\limits _{j=1} ^r (t - \alpha _j) \in {\mathbb Q}[t]
\]
for some $\alpha _j \in {\mathbb C}$ with $|\alpha _j| = \frac{1}{\sqrt{q}}$ for all $1 \leq j \leq r$.
If $\alpha _j$ is a real root of $P_C(t)$ then $\alpha _j = \frac{\varepsilon}{\sqrt{q}}$ with $\varepsilon = \pm 1$.
We claim that in the case of an even degree $r = 2m$, the zeta polynomial $P_C(t)$ is of the form
\begin{equation}   \label{EvenGeneric}
P_C(t) = a_{2m} \prod\limits _{i=1} ^m (t - \alpha _i)(t - \overline{\alpha _i})
\end{equation}
or of the form
\begin{equation}   \label{EvenNonGeneric}
P_C(t) = a_{2m} \left( t^2 - \frac{1}{q} \right) \prod\limits _{i=1} ^{m-1}  (t - \alpha _i)( t - \overline{\alpha _i}),
\end{equation}
while for an odd degree $r = 2m+1$ one has
\begin{equation}    \label{Odd}
P_C(t) = a_{2m+1} \left( t - \frac{\varepsilon}{\sqrt{q}} \right) \prod\limits _{i=1} ^m (t - \alpha _i)(t - \overline{\alpha _i})
\end{equation}
for some $\varepsilon \in \{ \pm 1\}$.
Indeed, if $\alpha _i \in {\mathbb C} \setminus {\mathbb R}$ is a complex, non-real root of $P_C(t) \in {\mathbb Q}[t] \subset {\mathbb R}[t]$ then $\overline{\alpha _i}  \neq \alpha _i$ is also a root of $P_C(t)$ and $P_C(t)$ is divisible by $(t - \alpha _i)(t - \overline{\alpha _i})$.
If $P_C(t) =0$ has three real roots $\alpha _1, \alpha _2, \alpha _3 \in \left \{ \frac{1}{\sqrt{q}}, - \frac{1}{\sqrt{q}} \right \}$, then at least  two of them coincide.
For $\alpha _1 = \alpha _2 = \frac{\varepsilon}{\sqrt{q}}$ one has $(t - \alpha _1)(t - \alpha _2) = (t - \alpha _1)(t - \overline{\alpha _1})$.
Thus, $P_C(t)$ has at most two real roots, which are not complex conjugate (or, equivalently, equal) to each other and $P_C(t)$ is of the form (\ref{EvenGeneric}), (\ref{EvenNonGeneric}) or (\ref{Odd}).

If $P_C(t)$ is of the form (\ref{EvenGeneric}), then $P_C(t) = a_{2m} \prod\limits _{i=1} ^m  \left( t^2 - 2 {\rm Re} ( \alpha _i) + \frac{1}{q} \right)$ and (\ref{DuursmaFunctionalEquation}) reads as
\begin{equation}   \label{EvenGenericFunctionalEquation}
P_{C^{\perp}} (t)  = a_{2m} \left[ \prod\limits _{i=1} ^m \left( \frac{1}{q} - 2 {\rm Re} ( \alpha _i) t + t^2 \right) \right ] q^{g-m} = P_C(t) q^{g-m},
\end{equation}
after multiplying each of the factors $\frac{1}{q^2t^2} - \frac{2 {\rm Re} ( \alpha _i)}{qt} + \frac{1}{q}$ by $qt^2$.
If $D_C(t)$ is Duursma's reduced polynomial of $C$ and $D_{C^{\perp}}(t)$ is Duursma's reduced  polynomial of $C^{\perp}$, then
\[
(1-t)(1 - qt) D_{C^{\perp}} (t) + t^{g^{\perp}} = P_{C^{\perp}} (t) = P_C (t) q^{g-m} = (1-t)(1 - qt) q^{g-m}  D_C(t) + q^{g-m} t^{g}
\]
implies that
\[
(1-t)(1-qt) [ D_{C^{\perp}} (t) - q^{g-m} D_C(t) ] = q^{g-m} t^g - t^{g^{\perp}}.
\]
Plugging in $t=1$, one concludes that $q^{g-m} =1$, whereas $g=m$.
As a result, $g + g^{\perp} = 2m = 2g$  specifies that $g = g^{\perp}$ and (\ref{EvenGenericFunctionalEquation})  yields $P_C(t) = P_{C^{\perp}} (t)$, which is equivalent to the formal self-duality of $C$.

If $P_C(t)$ is of the form (\ref{EvenNonGeneric}) then (\ref{DuursmaFunctionalEquation}) provides
\begin{equation}    \label{EvenNonGenericFunctionalEquation}
P_{C^{\perp}} (t) = a_{2m} \left( \frac{1}{q} - t^2 \right) \left[ \prod\limits _{i=1} ^{m-1}
\left( \frac{1}{q} - 2 {\rm Re} ( \alpha _i) t + t^2 \right) \right] q^{g-m} =  - P_C(t) q^{g-m}.
\end{equation}
Expressing by Duursma's reduced polynomials $D_C(t), D_{C^{\perp}}(t)$, one obtains
\begin{align*}
(1-t)(1 - qt) D_{C^{\perp}}(t)   + t^{g^{\perp}} = P_{C^{\perp}}(t) =  \\
 - P_C(t) q^{g-m} = - (1-t)(1 - qt) q^{g-m} D_C(t) - q^{g-m} t^g,
\end{align*}
whereas
\[
(1-t)(1-qt) [ D_{C^{\perp}}(t) + q^{g-m} D_C(t) ] = - t^{g^{\perp}} - q^{g-m} t^g.
\]
The substitution $t=1$ in the last equality of polynomials yields $-1 - q^{g-m} =0$, which is an absurd, justifying that a zeta polynomial $P_C(t)$, subject to the Riemann Hypothesis Analogue cannot be of the form (\ref{EvenNonGeneric}).

If $P_C(t)$ is of odd degree $ 2m+1$, then  (\ref{Odd})  and  (\ref{DuursmaFunctionalEquation}) yield
\begin{align*}
P_{C^{\perp}} (t)  =
 -  \varepsilon \sqrt{q} a_{2m+1} \left( t - \frac{\varepsilon}{\sqrt{q}} \right) \left[ \prod\limits _{i=1} ^m \left( \frac{1}{q} - 2 {\rm Re} ( \alpha _i) t + t^2 \right) \right] q^{g-m-1}   \\
  = - \varepsilon \sqrt{q} P_C(t) q^{g-m-1}
\end{align*}
after multiplying $\frac{1}{qt} - \frac{\varepsilon}{\sqrt{q}}$ by $- \frac{\varepsilon}{\sqrt{q}} qt$ and each
$\frac{1}{q^2t^2} - \frac{2 {\rm Re} ( \alpha _i)}{qt} + \frac{1}{q}$ by $qt^2$.
Expressing by Duursma's reduced polynomials
\begin{align*}
(1-t)(1-qt) D_{C^{\perp}}(t) + t^{g^{\perp}} = P_{C^{\perp}}(t) = - \varepsilon q^{g - m - \frac{1}{2} } P_C(t)  \\
=  - \varepsilon q^{g - m - \frac{1}{2} } (1-t)(1-qt)  D_C(t) - \varepsilon q^{g - m - \frac{1}{2}} t^g,
 \end{align*}
 one obtains
 \[
 (1-t)(1-qt) \left[ D_{C^{\perp}}(t) + \varepsilon q^{g - m - \frac{1}{2} } D_C(t) \right] = - t^{g^{\perp}} - \varepsilon q^{g - m - \frac{1}{2}} t^g.
 \]
The substitution $t=1$ implies $-1 - \varepsilon q^{g - m - \frac{1}{2} } =0$, which is an absurd, as far as $q^x =1$ if and only if $x=0$, while $g-m- \frac{1}{2}$ cannot vanish for integers $g,m$.
Thus, none zeta polynomial of odd degree satisfies the Riemann Hypothesis Analogue.

\end{proof}

\begin{corollary}   \label{BoundOnFieldCardinality}
If an   ${\mathbb F}_q$-linear code $C$ of $\dim _{{\mathbb F}_q} C =k$ and minimum distance $d$ satisfies the Riemann Hypothesis Analogue then the cardinality $q$ of the basic field  satisfies the upper bound
\[
q \leq \left( \sqrt[2g]{\binom{2k}{d} } +1 \right)^2.
\]
\end{corollary}

\begin{proof}

By Proposition \ref{RHAImpliesFSD}, if $C$ satisfies the Riemann Hypothesis Analogue then
\[
P_C(t) = a_{2g} \prod\limits _{j=1} ^q \left( t - \frac{e^{i \varphi _j}}{\sqrt{q}} \right) \left( t - \frac{e^{-i \varphi _j}}{\sqrt{q}} \right)
\]
for some $\varphi _j \in [0, 2 \pi )$.
The formal self-duality of $C$ is equivalent to the functional equation $P_C(t) = P_C \left( \frac{1}{qt} \right) q^g t^{2g}$ of the Hasse-Weil polynomial  of a function field of genus $g$ over ${\mathbb F}_q$ and implies that $a_{2g} = q^g a_0$.
Comparing the coefficients of $x^{2k-d}y^d$ in the expression
\[
\mathcal{W}_C(x,y) = a_0 \mathcal{M}_{2k,d}(x,y) + a_1 \mathcal{M}_{2k,d+1} (x,y) + \ldots + a_{2g} \mathcal{M}_{2k, d+2g} (x,y)
\]
of the homogeneous weight enumerator $\mathcal{W}_C(x,y)$ of $C$ by the homogeneous weight enumerators $\mathcal{M}_{2k,d+i}(x,y)$  of  MDS-codes of length $2k$ and minimum distance $d+i$, one concludes that $\mathcal{W}_C ^{(d)} = a_0 \mathcal{M}_{2k,d} ^{(d)} = a_0 (q-1) \binom{2k}{d}$.
Note that any word $c \in C$ of weight $d$ is a solution of a homogeneous linear system of rank $d-1$ in $d$ variables, as far as any $d-1$ columns of a parity check matrix   of $C$ are linearly independent.
Thus, there are exactly $q-1$ words of weight $d$ from $C$ with the same support as $c$.
If $\nu$ is the number of the $d$-tuples, supporting a word $c \in C$ of weight $d$ then $\mathcal{W}_C ^{(d)} = (q-1) \nu$ and
\[
a_0  = \frac{\nu}{\binom{2k}{d}}
\]
 is the probability for a $d$-tuple to support a word of weight $d$ from $C$.
Altogether, one obtains that
\begin{align*}
P_C(t) =
\frac{q^g \nu}{\binom{2k}{d}} \prod\limits _{j=1} ^q
\left( t - \frac{e^{i \varphi _j}}{\sqrt{q}} \right) \left( t - \frac{e^{-i \varphi _j}}{\sqrt{q}} \right) =  \\
\frac{\nu}{\binom{2k}{d}} \prod\limits _{j=1} ^q ( \sqrt{q}t - e^{i \varphi _j} )( \sqrt{q}t - e^{  - i \varphi _j}) =
\frac{\nu}{\binom{2k}{d}} \prod\limits (qt^2 - 2 \sqrt{q}t \cos \varphi _j t +1).
\end{align*}
In particular,
\[
1 = P_C(1) = \frac{\nu}{\binom{2k}{d}} \prod\limits _{j=1} ^q (q - 2 \sqrt{q} \cos \varphi _j +1).
\]
Bearing in mind that $\cos \varphi _j \in [ -1,1]$, one estimates
\[
q - 2 \sqrt{q} \cos \varphi _j + 1 \geq ( \sqrt{q} -1)^2
\]
and concludes that
\[
1 = \frac{\nu}{\binom{2k}{d}} \prod\limits _{j=1} ^q (q - 2 \sqrt{q} \cos \varphi _j +1) \geq \frac{\nu}{\binom{2k}{d}} (\sqrt{q} -1)^{2g}.
\]
As a result, there follows
\[
q \leq \left( \sqrt[2g]{\frac{\binom{2k}{d}}{\nu}} +1 \right)^2.
\]
By assumption, $C$ is of minimum distance $d$, so that $\nu \geq 1$ and
\[
\left( \sqrt[2g]{\frac{\binom{2k}{d}}{\nu}} +1 \right)^2 \leq \left( \sqrt[2g]{\binom{2k}{d}} +1 \right)^2.
\]

\end{proof}

\begin{proposition}   \label{FSD}
The following conditions are equivalent for a  linear code $C \subset {\mathbb F}_q^n$:

(i) $C$ is formally self-dual, i.e., the zeta polynomial $P_C(t)$ of $C$ satisfies the functional equation
\[
P_C(t) = P_C \left( \frac{1}{qt} \right) q^g t^{2g}
\]
of the Hasse-Weil polynomial of the function field of a curve of genus $g$ over ${\mathbb F}_q$;

(ii) Duursma's reduced polynomial $D_C(t) = \sum\limits _{i=0} ^{g + g^{\perp} -2} c_i t^{i}$ satisfies the functional equation
\begin{equation}   \label{Genus_G-1_Func_Eq}
D_C(t) = D_C \left( \frac{1}{qt} \right) q^{g-1} t^{2g-2}
\end{equation}
of the Hasse-Weil polynomial of the function field of a curve of genus $g-1$ over ${\mathbb F}_q$;

(iii)  the coefficients of Duursma's reduced polynomial $D_C(t) = \sum\limits _{i=0} ^{g  + g^{\perp} -2} c_i t^{i}$ of $C$ satisfy the equalities
\begin{equation}   \label{Realtions}
c_{g-1+i} = q^{i} c_{g-1 -i} \ \ \mbox{   for }   \ \ \forall 1 \leq i \leq g-1;
\end{equation}

(iv)  the dual code $C^{\perp} \subset {\mathbb F}_q^n$ of $C$ has dimension $\dim _{{\mathbb F}_q} C^{\perp} = \dim _{{\mathbb F}_q} C = k$, genus $g(C^{\perp}) = g(C) =g$ and the homogeneous weight enumerator of $C$ is
\begin{equation}   \label{FSDHWE}
\mathcal{W}_C (x,y) = \mathcal{M}_{2k, k+1} (x,y) + \sum\limits _{j=0} ^{g-1} c_{g-1-j} w_j (x,y),
\end{equation}
where
\begin{equation}   \label{GenericBinomial}
w_j (x,y) := (q-1) \binom{2k}{k+j} \left[ (x-y) ^{k+j} y^{k-j} + q^{j} (x-y) ^{k-j} y^{k+j} \right]
\end{equation}
 for   $ 1 \leq j \leq g-1$.
\begin{equation}   \label{ZeroBinomial}
w_0 (x,y) := (q-1) \binom{2k}{k} (x-y)^k y^k.
\end{equation}

(v) the dual code $C^{\perp} \subset {\mathbb F}_q^n$ of $C$ has dimension $\dim _{{\mathbb F}_q} C^{\perp} = \dim _{{\mathbb F}_q} C = k$, genus $g(C^{\perp}) = g(C) =g$ and the homogeneous weight enumerator
\begin{equation}   \label{ConstructionFSDHWE}
\mathcal{W}_C (x,y) = \mathcal{M}_{2k,k+1} (x,y) + \sum\limits _{w=d} ^{k-1} \mathcal{W}_C ^{(w)} \varphi _w (x,y) + \mathcal{W}_C ^{(k)} (x-y)^k y^k
\end{equation}
with
\begin{equation}   \label{ConstructionBinomial}
\varphi _w (x,y) := \sum\limits _{s=w} ^{k-1} \binom{2k-w}{s-w} \left[ (x-y)^{2k-s} y^s + q^{k-s} (x-y) ^s y^{2k-s} \right] + \binom{2k-w}{k} (x-y)^k y^k
\end{equation}
for $d \leq w \leq k-1$, so that $C$  can be obtained from   an MDS-code of the same length $2k$ and dimension $k$ by removing and adjoining appropriate words, depending explicitly on the numbers $\mathcal{W}_C ^{(d)}, \mathcal{W}_C ^{(d+1)}, \ldots , \mathcal{W}_C ^{(k)}$ of the codeword of $C$ of weight
 $\leq k = \dim _{{\mathbb F}_q} C$.

\end{proposition}

\begin{proof}

Towards $(i) \Rightarrow (ii)$, one  substitutes by $P_C(t) = (1-t)(1-qt) D_C(t) + t^g$ in (\ref{FuncEqFuncFieldGenusGCurveOverFQ}),  in order to obtain
\[
(1-t)(1-qt) D_C(t) + t^g =
 (qt-1)(t-1) \left[ D_C \left( \frac{1}{qt} \right) q^{g-1} t^{2g-2} \right] + t^g,
\]
whereas (\ref{Genus_G-1_Func_Eq}).

Conversely,  $(ii) \Rightarrow (i)$ is justified by
\begin{align*}
P_C(t) = (1-t)(1-qt) D_C(t) + t^g =  \\
= (t-1)(qt-1) \left[ D_C \left( \frac{1}{qt} \right) q^{g-1} t^{2g-2} \right] + t^g    \\
= \left[ \left( 1 - \frac{1}{t} \right) t \right]
\left[ \left( 1 - \frac{1}{qt} \right) qt \right] \left[ D_C \left( \frac{1}{qt} \right) q^{g-1} t^{2g-2} \right] + \frac{q^g t^{2g}}{q^g t^g}  \\
= \left[  \left( 1 - \frac{q}{qt} \right) \left( 1 - \frac{1}{qt} \right) D_C \left( \frac{1}{qt} \right) + \frac{1}{(qt)^g} \right] q^g t^{2g} =
P_C \left( \frac{1}{qt} \right) q^g t^{2g}.
\end{align*}
That proves the equivalence $(i) \Leftrightarrow (ii)$.

Towards $(ii) \Leftrightarrow (iii)$, note that the functional equation of $D_C(t)$ reads as
\begin{align*}
\sum\limits _{i=0} ^{2g-2} c_i t^{i} = D_C(t) = D_C \left( \frac{1}{qt} \right) q^{g-1} t^{2g-2} =
\left( \sum\limits _{i=0} ^{2g-2} \frac{c_i}{q^{i} t^{i}} \right) q^{g-1} t^{2g-2}   \\
=\sum\limits _{i=0} ^{2g-2} c_i q^{g-1-i} t^{2g-2-i} =
\sum\limits _{j=0} ^{2g-2} c_{2g-2-j} q^{-g+1+j} t^{j}.
\end{align*}
Comparing the coefficients of the left-most and the right-most side, one expresses the formal self-duality of $C$ by the relations
\[
c_j = q^{-g+1+j} c_{2g-2-j} \ \ \mbox{  for } \ \ \forall 0 \leq j \leq 2g-2.
\]
Let $i := g-1-j$, in order to express the above conditions in the form
\begin{equation}   \label{AllRelations}
c_{g-1+i} = q^{i} c_{g-1-i} \ \ \mbox{  for  }  \forall  -g+1 \leq i \leq g-1.
\end{equation}
For any $-g+1 \leq i \leq -1$ note that $c_{g-1+i} = q^{i} c_{g-1-i}$ is equivalent to $c_{g-1-i} = q^{-i} c_{g-1+i}$ and follows from (\ref{AllRelations})  with $1 \leq -i \leq g-1$.
In the case of $i=0$, (\ref{AllRelations}) holds trivially and (\ref{AllRelations}) amounts to (\ref{Realtions}).
That proves the equivalence of $(ii)$ with $(iii)$.

Towards $(iii) \Rightarrow (iv)$,  one  introduces a new variable $z:= x-y$ and expresses (\ref{WByD}) in the form
\begin{align*}
\mathcal{V}_C (y+z, y) := \mathcal{W}_C (y+z, y) - \mathcal{M}_{2k,k+1} (y+z,y) =
(q-1) \sum\limits _{i=0} ^{2g-2} c_i \binom{2k}{d+i} y^{d+i} z^{2k-d-i}  \\
 = (q-1) \sum\limits _{i=0} ^{g-1} c_i \binom{2k}{d+i} y^{d+i} z^{2k-d-i} +  (q-1) \sum\limits _{i=g} ^{2g-2} c_i \binom{2k}{d+i} y^{d+i} z^{2k-d-i}.
\end{align*}
Let us change the summation index of the first sum to $0 \leq j := g-1-i \leq g-1$,  put $1 \leq j:= i-g+1 \leq g-1$ in the second sum and make use of
 $d + g = k+1$, in order to obtain
\begin{equation}   \label{VFSDHWE}
\begin{split}
\mathcal{V}_C (y+z,y)  \\
 = (q-1) \sum\limits _{j=0} ^{g-1} c_{g-1-j}  \binom{2k}{k-j} y^{k-j} z^{k+j} +
(q-1) \sum\limits _{j=1} ^{g-1} c_{j+g-1}  \binom{2k}{k+j} y^{k+j} z^{k-j}.
\end{split}
\end{equation}
Extracting the term with $j=0$ from the first sum, one expresses
\begin{equation}   \label{GenHWE}
\begin{split}
\mathcal{V}_C (y+z,y) = (q-1) c_{g-1}  \binom{2k}{k} y^kz^k   \\
+  \sum\limits _{j=1} ^{g-1} (q-1) \binom{2k}{k+j} \left[ c_{g-1-j} y^{k-j} z^{k+j} + c_{g-1+j} y^{k+j} z^{k-j} \right]
\end{split}
\end{equation}
for an arbitrary ${\mathbb F}_q$-linear code $C \subset {\mathbb F}_q^n$.
If $C$ is formally self-dual, then plugging in by (\ref{Realtions}) in (\ref{GenHWE}) and making use of (\ref{GenericBinomial}), (\ref{ZeroBinomial}), one gets
\[
\mathcal{V}_C (y+z, y) = \sum\limits _{j=0} ^{g-1} c_{g-1-j} w_j (y+z,y).
\]
Substituting $z:= x-y$ and $\mathcal{V}_C (x,y) := \mathcal{W}_C(x,y) - \mathcal{M}_{2k,k+1} (x,y)$, one derives the equality  (\ref{FSDHWE}) for the homogeneous weight enumerator of a formally self-dual linear code $C \subset {\mathbb F}_q^{2k}$.

In order to justify that  (iv)   suffices for the formal self-duality of $C$, we use that  (\ref{FSDHWE})  with (\ref{GenericBinomial}) and (\ref{ZeroBinomial})   is equivalent to
\begin{equation} \label{FSD2}
\begin{split}
\mathcal{V}_C (y+z,y) = \sum\limits _{j=1} ^{g-1} c_{g-1-j} (q-1)  \binom{2k}{k+j} y^{k-j} z^{k+j}   \\
 + c_{g-1} (q-1) \binom{2k}{k} y^kz^k +
 \sum\limits _{j=1} ^{g-1} c_{g-1-j} (q-1) \binom{2k}{k+j} y^{k+j} z^{k-j}
\end{split}
\end{equation}
Comparing the coefficients of $y^{k+j} z^{k-j}$ with $1 \leq j \leq g-1$ from (\ref{GenHWE}) and (\ref{FSD2}), one concludes that
\[
c_{g-1+j} = c_{g-1-j} q^{j} \ \ \mbox{  for  } \ \ \forall 1 \leq j \leq g-1.
\]
These are exactly the relations (\ref{Realtions}) and imply the formal self-duality of $C$.

Towards $(iv) \Leftrightarrow (v)$, it suffices to  put $\mathcal{E}(x,y) := \sum\limits _{j=0} ^{g-1} c_{g-1-j} w_j (x,y)$ and  to derive  that $\mathcal{E}(x,y) = \sum\limits _{w=d} ^{k-1} \mathcal{W}_C ^{(w)} \varphi _w (x,y) + \mathcal{W}_C ^{(k)} (x-y) ^k y^k$.
More  precisely, introducing $i := g-1-j$, one expresses
\begin{align*}
\mathcal{E}(x,y) = \sum\limits _{i=0} ^{g-2} c_i (q-1) \binom{2k}{d+i} \left[ (x-y) ^{2k-d-i} y^{d+i} + q^{g-1-i} (x-y) ^{d+i} y^{2k-d-i} \right] \\
+ c_{g-1} (q-1) \binom{2k}{k} (x-y) ^k y^k.
\end{align*}
Plugging in by (\ref{DCByW1}) and exchanging the summation order,  one gets
\begin{align*}
\mathcal{E}(x,y) = \sum\limits _{w=d} ^{k-1} \sum\limits _{i=w-d} ^{g-2} \binom{2k-w}{d+i-w} \mathcal{W}_C ^{(w)}
 [ (x-y) ^{2k-d-i} y^{d+i} + q^{g-1-i} (x-y) ^{d+i} y^{2k-d-i} ]  \\
+ \sum\limits _{w=d} ^k \binom{2k-w}{k} \mathcal{W}_C ^{(w)} (x-y)^k y^k.
\end{align*}
Introducing $s:= d+i$ and extracting $\mathcal{W}_C ^{(w)}$ as coefficients, one  obtains
\[
\mathcal{E}(x,y) = \sum\limits _{w=d} ^{k-1} \mathcal{W}_C ^{(w)} \varphi _w (x,y) + \mathcal{W}_C ^{(k)} (x-y) ^k y^k.
\]

\end{proof}

Let $C \subset {\mathbb F}_q^n$ be an ${\mathbb F}_q$-linear code of genus $g$, whose dual $C^{\perp} \subset {\mathbb F}_q^n$ is of genus $g^{\perp}$.
In \cite{DL}, Dodunekov and Landgev introduce the near-MDS linear  codes $C$ as the ones with zeta polynomial $P_C(t) \in {\mathbb Q}[t]$ of degree
$\deg P_C(t) := g + g^{\perp} = 2$.
Thus,   $C$ is a near-MDS code if and only if  it has constant  Duursma's reduced polynomial $D_C(t) = c_0 \in {\mathbb Q}$.
Kim an Hyun prove in \cite{KH}) that a near-MDS code $C$ satisfies the Riemann Hypothesis Analogue exactly when
\[
 \frac{1}{(\sqrt{q}+1)^2} \leq c_0 \leq \frac{1}{(\sqrt{q} -1)^2}.
 \]
The next proposition characterizes the  formally-self-dual codes $C \subset {\mathbb F}_q^n$ of genus $2$, which satisfy the Riemann Hypothesis Analogue.
By Proposition \ref{FSD} (ii), $C$ is a formally self-dual linear code of genus $2$  exactly when its Duursma's reduced polynomial is
\[
D_C(t) = c_0 + c_1 t + q c_0 t^2
\]
for some $c_0, c_1 \in {\mathbb Q}$, $0 < c_0 < 1$.

\begin{proposition}   \label{RHAForQuadraticD}
A formally self-dual linear code $C \subset {\mathbb F}_q^{2k}$  with a quadratic Duursma's reduced polynomial $D_C(t) = c_0 + c_1 t + qc_0 t^2 \in {\mathbb Q}[t]$, $0 < c_0 < 1$ satisfies the Riemann Hypothesis Analogue if and only if
\begin{equation}   \label{Discriminant}
[(q+1) c_0 + c_1] ^2 \geq 4c_0,
\end{equation}
\begin{equation}  \label{Vertex}
q - 4 \sqrt{q} +1 \leq \frac{c_1}{c_0} \leq q + 4 \sqrt{q} +1,
\end{equation}
\begin{equation}   \label{EndValues}
c_1 \leq \min \left( \frac{1}{(\sqrt{q} -1) ^2} - 2 \sqrt{q} c_0, \, \frac{1}{( \sqrt{q} +1)^2} + 2 \sqrt{q} c_0 \right).
\end{equation}

\end{proposition}

\begin{proof}

According to (\ref{EvenGeneric}) from the proof of Proposition \ref{RHAImpliesFSD},  the zeta polynomial
\[
P_C(t) = (1-t)(1-qt) (qc_0 t^2 + c_1 t + c_0) + t^2
\]
satisfies the Riemann Hypothesis Analogue if and only if there exist $\varphi, \psi \in [0, 2 \pi )$ with
\[
P_C(t) = q^2 c_0 \left( t - \frac{e^{i \varphi}}{\sqrt{q}} \right) \left( t - \frac{e^{ - i \varphi}}{\sqrt{q}} \right)
\left( t - \frac{e^{i \psi}}{\sqrt{q}} \right) \left( t - \frac{e^{ - i \psi}}{\sqrt{q}} \right).
 \]
 Comparing the coefficients of $t$ and $t^2$ from $P_C(t)$, one expresses this condition by the  equalities
\begin{align*}
c_1 - (q+1) c_0 = - 2 \sqrt{q} c_0 [ \cos ( \varphi) + \cos ( \psi) ], \\
1 + 2q c_0 - (q+1) c_1 = 2q c_0 [1 + 2 \cos ( \varphi) \cos ( \psi) ].
\end{align*}
These  are equivalent to
\[
\cos ( \varphi) + \cos ( \psi) = \frac{(q+1) c_0 - c_1}{2 \sqrt{q} c_0}
\]
and
\[
\cos ( \varphi) \cos ( \psi) = \frac{1 - (q+1) c_1}{4qc_0}.
\]
In other words, the quadratic equation
\[
f(t) := t^2 + \frac{c_1 - (q+1) c_0}{2 \sqrt{q} c_0} t + \frac{1 - (q+1) c_1}{4qc_0} \in {\mathbb Q}[t]
\]
has roots $-1 \leq t_1 = \cos (\varphi)  \leq t_2 = \cos (\psi) \leq 1$.
This, in turn, holds exactly when the discriminant
\begin{equation}   \label{Discriminant1}
D(f) = \left[ \frac{c_1 - (q+1) c_0}{2 \sqrt{q} c_0} \right] ^2 - \frac{4 [ 1 - (q+1) c_1]}{4qc_0} \geq 0
\end{equation}
is non-negative, the vertex
\begin{equation}  \label{Vertex1}
-1 \leq \frac{(q+1) c_0 - c_1}{4 \sqrt{q} c_0} \leq 1
\end{equation}
belongs to the segment $[-1,1]$ and the values of $f(t)$ at the ends of this segment are non-negative,
\begin{equation}  \label{EndValues1}
f(1) \geq 0, \ \ f(-1) \geq 0.
\end{equation}
The equivalence of  (\ref{Discriminant1})  to  (\ref{Discriminant}) is straightforward.
Since $C$ is of minimum distance $d = k-1$ and $\mathcal{W}_C ^{(k-1)} = (q-1) \binom{2k}{k-1} c_0 \in {\mathbb N}$, the constant term  $c_0 >0$  of $D_C(t)$ is a positive rational number and one can   multiply (\ref{Vertex1})  by $ - 4 \sqrt{q} c_0  < 0$, add $(q+1) c_0$ to all the terms  and rewrite it  in the form
\[
(q - 4 \sqrt{q} +1) c_0 \leq c_1 \leq (q + 4 \sqrt{q} +1) c_0.
\]
Making use of $c_0 >0$, one observes that  the above inequalities are tantamount to  (\ref{Vertex}).
Finally,
\[
4q c_0 f(1) = 4q c_0 + 2 \sqrt{q} [ c_1 - (q+1) c_0 ] + 1 - (q+1) c_1 =
( - c_1 - 2 \sqrt{q} c_0)( \sqrt{q} -1) ^2 + 1 \geq 0
\]
 and
\[
4 q c_0 f( -1) = 4q c_0 - 2 \sqrt{q} [ c_1 - (q+1) c_0] + 1 - (q+1) c_1 =
(2 \sqrt{q} c_0 - c_1)( \sqrt{q} +1) ^2 +1 \geq 0
\]
can be expressed as  (\ref{EndValues}).

\end{proof}


\section{Duursma's reduced polynomial of a function field}

Let $F = {\mathbb F}_q(X)$ be the function field of a curve $X$ of genus $g$ over ${\mathbb F}_q$ and $h_g := h(F)$ be the class number of $F$, i.e., the number of the linear equivalence classes of the divisors of $F$ of degree $0$.
The present section introduces an additive decomposition of the Hasse-Weil polynomial $L_F(t) \in {\mathbb Z}[t]$ of $F$, which associates to $F$ a sequence
$\{ h_i \} _{i=1} ^{g-1}$ of virtual class numbers $h_i$ of function fields of curves of genus $i$ over ${\mathbb F}_q$.

\begin{lemma}   \label{EquivalentConditionsHasseWeilPolynomials}
The following conditions  are equivalent for a polynomial $L_g (t) \in {\mathbb Q} [t]$ of degree $\deg L_g(t) = 2g$:

$(i)$ $L_g(t)$ satisfies the functional equation
\[
L_g (t) = L_g \left( \frac{1}{qt} \right) q^g t^{2g}
\]
of the Hasse-Weil polynomial of the function field of a curve of genus $g$ over ${\mathbb F}_q$;

\[
(ii)  \hspace{4cm} L_{g-1} (t) := \frac{L_g (t) - L_g (1) t^g}{(1-t)(1-qt)}   \hspace{3cm}
\]
 is a polynomial with rational coefficients of degree $2g-2$, satisfying the functional equation
\[
L_{g-1} (t) = L_{g-1} \left( \frac{1}{qt} \right) q^{g-1} t^{2g-2}
\]
of the Hasse-Weil polynomial of the function field of a curve of genus $g-1$ over ${\mathbb F}_q$;

\[
(iii) \hspace{4cm} \quad L_g (t) = \sum\limits _{i=0} ^g h_i t^{i}  (1-t)^{g-i} (1-qt) ^{g-i}   \hspace{3cm}
\]
for some rational numbers $h_i \in {\mathbb Q}$.
\end{lemma}

\begin{proof}

Towards $(i) \Rightarrow (ii)$, let us note that the polynomial $M_g (t) := L_g (t) - L_g (1) t^g$ vanishes at $t=1$, so that it is divisible by $1-t$.
Further,
\[
M_g (t) = L_g (t) - L_g (1) t^g =
\left[ L_g \left( \frac{1}{qt} \right) - \frac{L_g (1)}{q^gt^g} \right] q^g t^{2g} =
M_g \left( \frac{1}{qt} \right) q^g t^{2g}
\]
  satisfies the functional equation  of the Hasse-Weil polynomial of the function field of a curve of genus $g$ over ${\mathbb F}_q$.
In particular, $M_g \left( \frac{1}{q} \right) = M_g (1) \frac{q^g}{q^{2g}} =0$ and $M_g (t)$ is divisible by the linear polynomial
 $q \left(  \frac{1}{q} - t \right) = 1 - qt$, which is relatively prime to $1-t$ in ${\mathbb Q}[t]$.
As a result,
\[
L_{g-1} (t) := \frac{M_g (t)}{(1-t)(1-qt)} \in {\mathbb Q}[t]
\]
is a polynomial of degree $\deg L_{g-1} (t) = 2g-2$.
Straightforwardly,
\begin{align*}
L_{g-1} \left( \frac{1}{qt} \right) q^{g-1} t^{2g-2} =
 \left[ M_g \left( \frac{1}{qt} \right)  : \left( 1 - \frac{1}{qt} \right) \left( 1 - \frac{1}{t} \right) \right] \\
= \frac{M_g (t)}{qt^2} : \frac{(qt-1)(t-1)}{qt^2} =
 \frac{M_g (t)}{(1-t)(1-qt)} = L_{g-1} (t)
 \end{align*}
satisfies the functional equation of the Hasse-Weil polynomial of the function field of a curve of genus $g-1$ over ${\mathbb F}_q$.

The implication $(ii) \Rightarrow (i)$ follows from the functional equation of $L_{g-1} (t)$, applied to $L_g (t) = (1-t)(1 - qt) L_{g-1} (t) + L_g (1) t^g$.
Namely,
\begin{align*}
L_g \left( \frac{1}{qt} \right) q^g t^{2g} \\
= \left[ \left( 1 - \frac{1}{qt} \right) qt \right] \left[ \left( 1 - \frac{1}{t} \right) t \right]
\left[ L_{g-1} \left( \frac{1}{qt} \right) q^{g-1} t^{2g-2} \right] + \frac{L_g (1)}{q^g t^g} q^g t^{2g} \\
= (qt-1)(t-1) L_{g-1} (t) + L_g (1) t^g  \\
= (1-t)(1-qt) L_{g-1} (t) + L_g (1) t^g = L_g (t).
\end{align*}

We derive $(i) \Rightarrow (iii)$ by an induction on $g$, making use of $(ii)$.
More precisely, for $g=1$ one has $L_0 (t) := \frac{L_1(t) - L_1 (1) t}{(1-t)(1 -qt)} \in {\mathbb Q}[t]$ of degree $\deg L_0 (t) =0$ or $L_0 \in {\mathbb Q}$.
Then
\[
L_1 (t) =  (1-t)(1 -qt) L_0 + L_1 (1) t =
\sum\limits _{i=0} ^1 h_i t^{i} (1-t) ^{1-i} (1-qt) ^{1-i}
\]
with $h_0 := L_0 \in {\mathbb Q}$ and $h_1 := L_1 (1) \in {\mathbb Q}$.
In the general case, $(ii)$ provides a polynomial
\[
L_{g-1} (t) := \frac{L_g (t) - L_g (1) t^g}{(1-t)(1 - qt)},
\]
subject to the functional equation
\[
L_{g-1} (t) = L_{g-1} \left( \frac{1}{qt} \right) q^{g-1} t^{2g-2}
\]
 of the Hasse-Weil polynomial of the function field of a curve of genus $g-1$ over
 ${\mathbb F}_q$.
By the  inductional hypothesis, there exist $h'_i \in {\mathbb Q}$, $0 \leq i \leq g-1$ with
\[
L_{g-1} (t) = \sum\limits _{i=0} ^{g-1} h'_i t^{i} (1-t) ^{g-1-i} (1-qt) ^{g-1-i}.
\]
Then
\[
L_g (t) = (1-t)(1 - qt) L_{g-1} (t) + L_g (1) t^g =
\sum\limits _{i=0} ^g h_i t^{i} (1-t)^{g-i} (1 - qt) ^{g-i}
\]
with $h_i := h'_i \in {\mathbb Q}$ for $0 \leq i \leq g-1$ and $h_g := L_g (1) \in {\mathbb Q}$ justifies $(i) \Rightarrow  (iii)$.

Towards  $(iii) \Rightarrow (i)$, let us  assume  that   $L_g(t) = \sum\limits _{i=0} ^g h_i t^{i} (1-t)^{g-i} (1 - qt) ^{g-i}$.
 Then
\begin{align*}
L \left( \frac{1}{qt} \right) q^g t^{2g} = \left[ \sum\limits _{i=0} ^g \frac{h_i}{q^{i} t^{i}} \left( 1 - \frac{1}{qt} \right) ^{g-i} \left( 1 - \frac{1}{t} \right) ^{g-i} \right] q^g t^{2g}  \\
= \sum\limits _{i=0} ^g \left[ \frac{h_i}{q^{i} t^{i}} q^{i} t^{2i} \right] \left[ \left( 1 - \frac{1}{qt} \right) qt \right] ^{g-i}
 \left[ \left( 1 - \frac{1}{t} \right) t \right] ^{g-i}   \\
= \sum\limits _{i=0} ^g h_i t^{i} (qt-1) ^{g-i} (t-1) ^{g-i} = L_g(t)
 \end{align*}
 satisfies the functional equation of the Hasse-Weil polynomial of the function field of a curve of genus $g$ over ${\mathbb F}_q$.

\end{proof}

\begin{proposition}   \label{HasseWeilDecomposition}
Let $F = {\mathbb F}_q(X)$ be the function field of a smooth irreducible curve $X / {\mathbb F}_q \subset {\mathbb P}^N ( \overline{{\mathbb F}_q})$ of genus $g$, defined  over ${\mathbb F}_q$, with $h(F)$ linear equivalence classes of divisors of degree $0$, $\mathcal{A}_i$ effective divisors  of degree $i \geq 0$,
Hasse-Weil polynomial $L_F(t) \in {\mathbb Q}[t]$ and Duursma's reduced polynomial $D_F (t) \in {\mathbb Q}[t]$, defined by the equality
\[
L_F (t) = (1-t)(1 - qt) D_F (t) + h(F) t^g.
\]
Then:

(i) $D_F(t) = \sum\limits _{i=0} ^{g-2} \mathcal{A}_i ( t^{i} + q ^{g-1-i} t^{2g-2-i} ) + \mathcal{A}_{g-1} t^{g-1} \in {\mathbb Z}[t]$
is a polynomial with integral coefficients, which is  uniquely determined by   $\mathcal{A}_0 =1, \mathcal{A}_1, \ldots , \mathcal{A}_{g-1}$;

(ii)  the equality
\begin{equation}  \label{ReducedZetaFunction}
\frac{D_F(t)}{(1-t)(1-qt)} = \sum\limits _{i=0} ^{\infty} \mathcal{B}_i t^{i}
\end{equation}
of formal power series of $t$  holds for
\begin{equation}    \label{LowerBI}
\mathcal{B} _i = \sum\limits _{j=0} ^{i} \mathcal{A}_j \left( \frac{q^{i-j+1} -1}{q-1} \right)
\end{equation}
for $0 \leq i \leq g-1$,
\begin{equation}    \label{MiddleBI}
\mathcal{B}_i = \sum\limits _{j=0} ^{g-1} \mathcal{A}_j \left( \frac{q^{i-j+1}-1}{q-1} \right) + \sum\limits _{j=g} ^{i} \mathcal{A}_{2g-2-j}
\left( \frac{q^{i-g+2}-q^{j-g+1}}{q-1} \right)
\end{equation}
for $g \leq i \leq 2g-3$,
\begin{equation}    \label{HigherBI}
\mathcal{B}_i = D_F(1) \left( \frac{q^{i-g+2}-1}{q-1} \right)
\end{equation}
for $i \geq 2g-2$;

(iii) the natural numbers $\mathcal{B}_i$, $i \geq 0$ from (ii) satisfy the relations
\begin{equation}    \label{ReducedHalfOfSpecialDivisors}
\mathcal{B}_i = q^{i-g+2} \mathcal{B} _{2g-4-i} + D_F(1) \left( \frac{q^{i-g+2} -1}{q-1} \right) \ \ \mbox{  for  } \ \  \forall g-1 \leq i \leq 2g-4;
\end{equation}
\begin{equation}   \label{ReducedNumberEffectiveNonSpecialDivisors}
\mathcal{B}_i = D_F(1) \left( \frac{q^{i-g+2}-1}{q-1} \right) \ \ \mbox{  for  } \ \ \forall i \geq 2g-3.
\end{equation}

(iv) the number $h(F)$ of the linear equivalence classes of the divisors of $F$ of degree $0$ satisfies the inequilities
\[
(\sqrt{q}-1)^{2g} \leq h(F) \leq (\sqrt{q}+1)^{2g}
\]
\end{proposition}

\begin{proof}

(i) By Theorem 4.1.6. (ii) and Theorem 4.1.11 from \cite{NX}, the Hasse-Weil zeta function of $F$ is the generating function
\[
Z_F (t) = \frac{L_F(t)}{(1-t)(1-qt)} = \sum\limits _{j=0} ^{\infty} \mathcal{A}_j t^{j}
\]
of the sequence $\{ \mathcal{A}_i \} _{i=0} ^{\infty}$.
According to Lemma  \ref{EquivalentConditionsHasseWeilPolynomials} and $L_F(1) = h(F)$,
\[
D_F(t) := \frac{L_F(t) - h(F)t^g}{(1-t)(1-qt)}
\]
  is a polynomial  of $\deg D_F (t) = 2g-2$, subject to the functional equation of the Hasse-Weil polynomial of the function field of a curve of genus $g-1$ over ${\mathbb F}_q$.
Thus,
\begin{equation}    \label{HWFunction}
Z_F (t) = D_F(t) + \frac{h(F) t^g}{(1-t)(1-qt)} = \sum\limits _{j=0} ^{\infty} \mathcal{A}_j  t^{j}.
\end{equation}
 Let  $l(G)$ is the dimension of the space $H^0 (X, \mathcal{O}_X(G))$ of the global holomorphic sections of the line bundle $\mathcal{O}_X(G) \rightarrow X$, associated with a divisor $G \in {\rm Div} (F)$.
    Riemann-Roch Theorem asserts that
\[
l(G ) =  l( K_X - G ) + \deg (G) -g+1
\]
for a canonical divisor $K_X$ of $X$.
For any $j \geq g-1$, suppose that  $G_1, \ldots , G_{h(F)} \in {\rm Div } (F)$  is a complete set of representatives of the linear equivalence classes of the  divisors of $F$ of degree $j$.
Then
\begin{equation}    \label{AJFirstPart}
\mathcal{A}_j
  = \sum\limits _{\nu=1} ^{h(F)} \frac{q^{l(G_{\nu})} -1}{q-1}
   =  q^{j-g+1} \sum\limits _{\nu=1} ^{h(F)} \left( \frac{q^{l(K_Y - G_{\nu})} -1}{q-1} \right) + h(F) \left( \frac{q^{j-g+1} -1}{q-1} \right)
\end{equation}
for $ g \leq j \leq 2g-2$  and
\begin{equation}    \label{NumberEffectiveNonSpecialDivisors}
\mathcal{A}_j = h(F) \left( \frac{q^{j-g+1} -1}{q-1} \right) \ \ \mbox{  for  } \ \ \forall j \geq 2g-1.
\end{equation}
 Note that $K_Y - G_1, \ldots , K_Y - G_{h(F)}$  is a complete set of representatives of the linear equivalence classes of the divisors of $F$ of degree
 $2g -2-j$, so that
\begin{equation}    \label{AJSecondPart}
\mathcal{A}_{2g-2-j} = \sum\limits _{\nu =1} ^{h(F)} \frac{q ^{l(K_Y - G_{\nu})} -1}{q-1}.
\end{equation}
Plugging in by (\ref{AJSecondPart}) in (\ref{AJFirstPart}),  one obtains
\begin{equation}     \label{HalfOfSpecialDivisors}
\mathcal{A}_j =  q^{j-g+1} \mathcal{A} _{2g-2-j} +   h(F) \left( \frac{q^{j-g+1} -1}{q-1} \right) \ \ \mbox{  for  } \ \ g \leq j \leq 2g-2,
\end{equation}
whereas
\[
Z_F (t) = \sum\limits _{j=0} ^{g-1} \mathcal{A}_j t^{j} + \sum\limits _{j=g} ^{2g-2} q^{j-g+1} \mathcal{A}_{2g-2-j} t^{j} +
h(F) \sum\limits _{j=g} ^{\infty} \left( \frac{q^{j-g+1} -1}{q-1} \right)  t^{j},
\]
Putting $i := 2g-2-j$ in the second sum and $i := j-g$ in the third sum, one expresses
\begin{align*}
Z_F (t) = \sum\limits _{i=0} ^{g-2} \mathcal{A}_i ( t^{i} + q^{g-1-i} t^{2g-2-i} ) + \mathcal{A}_{g-1} t^{g-1} \\
+ h(F)  \left[  \frac{qt^g}{q-1} \left( \sum\limits _{i=0} ^{\infty} q^{i} t^{i} \right) -
\frac{t^g}{q-1} \left( \sum\limits _{i=0} ^{\infty} t^{i} \right) \right],
\end{align*}
Summing up the geometric progressions
\[
\sum\limits _{i=0} ^{\infty} q^{i} t^{i} = \frac{1}{1 - qt}, \ \
\sum\limits _{i=0} ^{\infty} t^{i} = \frac{1}{1-t},
\]
one derives
\[
Z_F (t) = \sum\limits _{i=0} ^{g-2} \mathcal{A}_i ( t^{i} + q ^{g-1-i} t^{2g-2-i} ) + \mathcal{A}_{g-1} t^{g-1}
+ h(F) \frac{t^g}{(1-t)(1 - qt)},
\]
whereas
\[
D_F (t) = \sum\limits _{i=0} ^{g-2} \mathcal{A}_i ( t^{i} + q^{ g-1-i} t^{2g-2-i} ) + \mathcal{A}_{g-1} t^{g-1}.
\]
In particular, $D_F(t) \in {\mathbb Z}[t]$ has integral coefficients.

(ii) Let us expand
\[
\frac{1}{1-t} = \sum\limits _{i=0} ^{\infty} t^{i}, \quad \frac{1}{1 - qt} = \sum\limits _{i=0} ^{\infty} q^{i} t^{i}
\]
as sums of geometric progressions and note that
\[
\frac{1}{(1-t)(1-qt)} = \sum\limits _{i=0} ^{\infty} (1 + q + \ldots + q^{i}) t^{i} =
\sum\limits _{i=0} ^{\infty} \left( \frac{q^{i+1} -1}{q-1} \right) t^{i}.
\]
Then represent Duursma's reduced polynomial in the form
\begin{equation}     \label{DIsDeterminedByAi}
D_F(t) = \sum\limits _{j=0} ^{g-1} \mathcal{A}_j t^{j} + \sum\limits _{j=g} ^{2g-2} \mathcal{A}_{2g-2-j} q^{j-g+1} t^{j}.
\end{equation}
Now, the comparison of the coefficients of $t^{i}$, $i \geq 0$ from the  left hand side and the right hand side of   (\ref{ReducedZetaFunction}) provides (\ref{LowerBI}), (\ref{MiddleBI}) and
\[
\mathcal{B} _i = \sum\limits _{j=0} ^{g-1} \mathcal{A}_j \left( \frac{q^{i-j+1} -1}{q-1} \right) +
\sum\limits _{j=g} ^{2g-2} \mathcal{A}_{2g-2-j}  q^{j-g+1} \left( \frac{q^{i-j+1}-1}{q-1} \right) \ \ \mbox{  for  } i \geq 2g-2.
\]
The last formula can be expressed in the form
\begin{align*}
\mathcal{B}_i  = \\
= \frac{q^{i+1}}{q-1} \left( \sum\limits _{j=0} ^{q-1} \mathcal{A}_j q^{-j} + \sum\limits _{j=g} ^{2g-2} \mathcal{A}_{2g-2-j} q^{j-g+1} q^{-j} \right)
- \frac{1}{q-1} \left( \sum\limits _{j=0} ^{g-1} \mathcal{A}_j + \sum\limits _{j=g} ^{2g-2} \mathcal{A}_{2g-2} q^{j-g+1} \right) \\
= \frac{q^{i+1}}{q-1} D_F  \left( \frac{1}{q} \right) - \frac{1}{q-1} D_F  (1).
\end{align*}

According to Lemma \ref{EquivalentConditionsHasseWeilPolynomials} $(i) \Rightarrow (ii)$, Duursma's reduced polynomial of $F$ satisfies the functional equation $D_F(t) = D_F \left( \frac{1}{qt} \right) q^{g-1} t^{2g-2}$.
In particular,  $D_F(1) = D_F \left( \frac{1}{q} \right) q^{g-1}$ and there follows (\ref{HigherBI}).

(iii)  Due to $\mathcal{A}_i \geq 0$ for $\forall i \geq 0$,  $\mathcal{B}_i$ are sums of non-negative integers.
Moreover, $\mathcal{B}_i \geq \mathcal{A}_i \left( \frac{q^{i+1}}{q-1} \right) \geq \mathcal{A}_0 = 1 >0$ for $\forall i \geq 0$ reveals that all $\mathcal{B}_i$ are natural numbers.
Towards (\ref{ReducedHalfOfSpecialDivisors}), let us introduce the polynomial $\psi (t) := \sum\limits _ {j=0} ^{g-2} \mathcal{A}_j t^{j} \in {\mathbb Z}[t]$ and express
\begin{align*}
D_F (t)
 = \sum\limits _{j=0} ^{g-2} \mathcal{A}_j t^{j} + q^{g-1} t^{2g-2} \left[ \sum\limits _{j=0} ^{g-2} \mathcal{A}_j (qt) ^{-j} \right] + \mathcal{A}_{g-1} t^{g-1} \\
  = \psi (t) + \psi \left( \frac{1}{qt} \right) q^{g-1} t^{2g-2} + \mathcal{A}_{g-1} t^{g-1}.
 \end{align*}
 In particular,
 \begin{equation}    \label{DFAt1}
 D_F(1) = \psi (1) + \psi \left( \frac{1}{q} \right) q^{g-1} + \mathcal{A}_{g-1}.
 \end{equation}
Straightforwardly,
\begin{align*}
\mathcal{B}_{g-1} - q \mathcal{B}_{g-3}    \\
= \frac{q^g}{q-1} \left( \sum\limits _{j=0} ^{g-2} \mathcal{A}_j q^{-j} \right) - \frac{1}{q-1} \left( \sum\limits _{j=0} ^{g-2} \mathcal{A}_j \right) +
\mathcal{A}_{g-1} -  \\
- \frac{q^{g-1}}{q-1} \left( \sum\limits _{j=0} ^{g-2} \mathcal{A}_j q^{-j} \right) + \frac{q}{q-1} \left( \sum\limits _{j=0} ^{g-2} \mathcal{A}_j \right)  \\
=   \psi \left( \frac{1}{q} \right) q^{g-1} + \psi (1)  + \mathcal{A}_{g-1} = D_F(1).
\end{align*}
That proves (\ref{ReducedHalfOfSpecialDivisors}) for $i = g-1$.
In the case of   $g \leq i \leq 2g-4$ note that  $0 \leq 2g-4-i \leq g-4$ and
\begin{align*}
(q-1) (\mathcal{B}_i - q^{i-g+2} \mathcal{B}_{2g-4-i} )  \\
= \sum\limits _{j=0} ^{g-1} \mathcal{A}_j (q^{i-j+1}-1) + \sum\limits _{j=g} ^{i} \mathcal{A}_{2g-2-j} (q^{i-g+2} - q^{j-g+1}) -
\sum\limits _{j=0} ^{2g-4-i} \mathcal{A}_j (q^{g-1-j} - q^{i-g+2}).
\end{align*}
Changing the summation index of the second sum to $s:=  2g-2-j$, one obtains
\begin{align*}
(q-1) ( \mathcal{B} _i - q^{i-g+2} \mathcal{B}_{2g-4-i} ) \\
= q^{i+1} \left( \sum\limits _{j=0} ^{g-1} \mathcal{A}_j q^{-j} \right) - \left( \sum\limits _{j=0} ^{g-1} \mathcal{A}_j \right)
+ q^{i-g+2} \left( \sum\limits _{s=2g-2-i} ^{g-2} \mathcal{A}_s \right)  \\
 - q^{g-1} \left( \sum\limits _{s= 2g-2-i} ^{g-2} \mathcal{A}_s q^{-s} \right)
- q^{g-1} \left( \sum\limits _{j=0} ^{2g-4-i} \mathcal{A}_j q^{-j} \right) + q^{i-g+2} \left( \sum\limits _{j=0} ^{2g-4-i} \mathcal{A}_j \right).
\end{align*}
An appropriate grouping of the sums yields
\begin{align*}
(q-1) ( \mathcal{B}_i - q^{i-g+2} \mathcal{B}_{2g-4-i}  )  \\
=  \psi \left( \frac{1}{q} \right) q^{i+1} +  \mathcal{A}_{g-1} q^{i-g+2}  - \psi (1)  - \mathcal{A}_{g-1} +  \psi (1) q^{i-g+2}
 -  \psi \left( \frac{1}{q} \right) q^{g-1}  \\
= (q^{i-g+2} -1) \left[ \psi (1) + \psi \left( \frac{1}{q} \right)  q^{g-1}  + \mathcal{A}_{g-1} \right] = D_F(1) (q^{i-g+2} -1).
\end{align*}
That justifies (\ref{ReducedHalfOfSpecialDivisors}).

Note that  (\ref{ReducedNumberEffectiveNonSpecialDivisors}) with $i \geq 2g-2$ coincides with (\ref{HigherBI}).
In the case of $i = 2g-3$,
\[
(q-1)  \mathcal{B}_{2g-3} = \sum\limits _{j=0} ^{g-1} \mathcal{A}_j (q^{2g-2-j} -1) + \sum\limits _{s=1} ^{g-2} \mathcal{A}_s  (q^{g-1}- q^{g-1-s}),
\]
after changing the summation index of the second sum to $s:= 2g-2-j$.
Then
\begin{align*}
(q-1) \mathcal{B}_{2g-3}  \\
 = q^{2g-2} \left( \sum\limits _{j=0} ^{g-2} \mathcal{A}_j q^{-j} \right)   - \left( \sum\limits _{j=0} ^{g-2} \mathcal{A}_j \right)
 +   \mathcal{A}_{g-1} (q^{g-1} -1) +  \\
  + q^{g-1} \left( \sum\limits _{j=0} ^{g-2} \mathcal{A}_j \right) - q^{g-1} \left( \sum\limits _{j=0} ^{g-2} \mathcal{A}_j q^{-j} \right)  \\
 = (q^{g-1} -1) \left[ \psi (1) + \psi \left( \frac{1}{q} \right) q^{g-1} + \mathcal{A}_{g-1} \right] = D_F(1) (q^{g-1} -1),
 \end{align*}
 which is tantamount  to (\ref{ReducedNumberEffectiveNonSpecialDivisors}) with $i  = 2g-3$.

(iv) By the Hasse-Weil Theorem, all the roots of $L_F(t)$ belong to the circle
$S \left( \frac{1}{\sqrt{q}} \right) = \left \{ z \in {\mathbb C} \ \ \vert \ \  |z| = \frac{1}{\sqrt{q}} \right \}$.
The proof of Proposition \ref{RHAImpliesFSD} specifies that
\[
L_F(t) = a_{2g} \prod\limits _{j=1} ^g \left( t - \frac{e^{i \varphi _j}}{\sqrt{q}} \right) \left( t - \frac{e^{-i \varphi _j}}{\sqrt{q}} \right)
\]
for some $\varphi _j \in [0, 2 \pi )$.
The functional equation $L_F(t) = L_F \left( \frac{1}{qt} \right) q^g t^{2g}$ implies that $a_{2g} = q^g a_0$.
Combining with $a_0 = L_F(0) =1$, one gets
\begin{align*}
L_F(t) = \prod\limits _{j=1} ^g ( \sqrt{q}t - e^{i \varphi _j})(\sqrt{q}t - e^{- i \varphi _j}) =
\prod\limits _{j=1} ^g (qt^2 - 2 \sqrt{q} \cos \varphi _j t +1).
\end{align*}
The substitution  $t=1$ provides
\[
h(F) = L_F(1) = \prod\limits _{j=1} ^g (q - 2 \sqrt{q} \cos \varphi _j +1).
\]
However, $\cos \varphi _j \in [ -1, 1]$ requires
\[
(\sqrt{q}-1)^2 \leq q - 2 \sqrt{q} \cos \varphi _j +1 \leq ( \sqrt{q} +1)^2,
\]
whereas
\[
( \sqrt{q}-1)^{2g} \leq h(F) = L_F(1) = \prod\limits _{j=1} ^g (q - 2 \sqrt{q} \cos \varphi _j +1) \leq ( \sqrt{q} +1)^{2g}.
\]

\end{proof}

\newpage

 \end{document}